\documentclass[12pt]{article}
\usepackage{amscd,amssymb,amsmath,latexsym,enumerate}
\usepackage[mathscr]{euscript}
\usepackage{graphicx}
\usepackage[final]{pdfpages}
\usepackage{mathrsfs}
\usepackage{epsfig}
\usepackage{verbatim}
\usepackage{color}
\usepackage{amsfonts,amsthm,hyperref,mathrsfs,ulem,tikz}
\usepackage{tikz}
\usepackage{tikz-cd}
\usetikzlibrary{calc}
\usepackage{float}
\usepackage{todonotes}
\usepackage{multicol}
\usepackage{mathtools}
\usepackage[T1]{fontenc}

\usetikzlibrary{arrows,patterns}

\usepackage{xcolor}
\definecolor{MyBlue}{cmyk}{1,0.13,0,0.63}
\definecolor{MyGreen}{cmyk}{0.91,0,0.88,0.52}
\newcommand{\mylinkcolor}{MyBlue}
\newcommand{\mycitecolor}{MyGreen}
\newcommand{\myurlcolor}{black}

\usepackage{hyperref}
\hypersetup{%
  bookmarksnumbered=true,bookmarksopen=false,%
  plainpages=false,
  linktocpage=true,%
  colorlinks=true,breaklinks=true,%
  linkcolor=\mylinkcolor,citecolor=\mycitecolor,urlcolor=\myurlcolor,%
  pdfpagelayout=OneColumn,%
  pageanchor=true,%
}

\title{Logarithmic transport at topological phase transitions \\ of one-dimensional chiral quantum systems}

\author{Dragan Markovi\'c, Hermann Schulz-Baldes
\\
\\
{\small Department Mathematik, Friedrich-Alexander-Universit\"at Erlangen-N\"urnberg, Germany}
}

\date{ }

\newtheorem{theo}{Theorem}

\newtheorem{proposi}[theo]{Proposition}
\newtheorem{lemma}[theo]{Lemma}
\newtheorem{coro}[theo]{Corollary}
\newtheorem{rem}[theo]{Remark}

\newcommand{\CM}{{\mathbb C}}
\newcommand{\NM}{{\mathbb N}}
\newcommand{\RM}{{\mathbb R}}

\newcommand{\ZM}{{\mathbb Z}}

\newcommand{\EM}{{\mathbb E}}

\newcommand{\PM}{{\mathbb P}}

\newcommand{\Ee}{{\cal E}}

\newcommand{\Oo}{{\cal O}}

\newcommand{\Tt}{{\cal T}}

\newcommand{\Nn}{{\cal N}}

\newcommand{\Hh}{{\cal H}}

\newcommand{\spec}{{\mbox{\rm spec}}}

\newcommand{\one}{{\bf 1}}

\newcommand{\sgn}{\mbox{\rm sgn}}

\newcommand{\argmax}{{\mbox{\rm argmax}}}
\newcommand{\argmin}{{\mbox{\rm argmin}}}

\newcommand{\bsm}{\left(\begin{smallmatrix}} 
\newcommand{\esm}{\end{smallmatrix}\right)}  
\definecolor{GR}{rgb}{.35,.7,.35}

\begin{document}

\maketitle

\begin{abstract}
The random hopping Hamiltonian is a toy model for a disorder-driven topological phase transitions in one-dimensional chiral Hamiltonians. It has a vanishing Lyapunov exponent and a Dyson peak in the density of states at zero energy. This work shows how the quantum dynamics leads to a logarithmic growth of the moments of the position operator, similar as in classical Sinai diffusion. The mechanism at the origin of this phenomenon is that the eigenfunction of the smallest eigenvalue is well-approximated by the exponential of a classical random walk, leading to two essentially independent subexponential localization centers.

\vspace{.1cm}


\end{abstract}




\section{Introduction and resum\'e}

Topological phases of non-interacting chiral Fermions in odd space dimensions can be distinguished by topological invariants given by (possibly higher) winding numbers. For disordered systems of this type, there are non-commutative versions of the winding number which due to index theorems remain integer-valued \cite{MSHP,PS}. By definition, a topological phase transition is a jump of this integer invariant under the change of an external parameter.  If this parameter is the disorder strength, one speaks of a disorder-driven topological phase transition. By general principles, the quantum dynamics at such a transition point cannot be Anderson localized \cite{PS,SSt1}. The best known example of this type (in a system without chiral symmetry) is the bending of the Landau level under the addition of a random potential, but here the focus is on quasi-one-dimensional chiral models like the disordered Su-Schrieffer-Heeger model \cite{SSH} for which the disorder-driven phase transition was first studied in \cite{MSHP}. In particular, it can be characterized by a vanishing of the Lyapunov exponent at zero energy \cite{MSHP,DSS2} and the existence of a Dyson peak in the density of states \cite{DSS}. The toy model for the transition point is the widely studied random hopping model. Its zero energy state and its transmission properties have been studied in the physics literature for a long time, {\it e.g.} \cite{TC,SE,ITA}. On another front, it has been argued in \cite{MHMD} that these transitions correspond to a quantum critical point that can be studied by a strong-disorder renormalization group as introduced for the analysis of spin chains \cite{Fis}, but later on transposed to many other models of statistical physics \cite{BF,McK,IM}. Partially based on this insight, the work \cite{BAK} argues that the quantum diffusion for such a model at the topological phase transition is very slow, namely moments of the position grow only logarithmically in time -- just as for classical Sinai diffusion in a random environment.  Of course, such logarithmic behavior is notoriously hard to see in numerical studies ({\it e.g.} \cite{PaS}), but there are experimental realizations showing that there is at least no Anderson localization at the transition points \cite{BSS} and also the effect on the ac conductivity has been studied \cite{KHQ}. This paper proves rigorously a quantitative logarithmic lower bound on the quantum dynamics on appropriate time and length scales. For sake of simplicity, the study is carried out for the random hopping model. 

\vspace{.2cm}

Before explaining the mathematical set-up and results in more detail, let us offer another perspective on the results of this paper, notably from the point of view of quantum dynamics in one dimension. As already pointed out, the topological phase transition can be characterized by a vanishing of the Lyapunov exponent $\gamma$ at zero energy, and furthermore also by the appearance of a Dyson peak \cite{Dys} in the density of states $\Nn$. More precisely, there exist computable positive constants $C$ and $C'$ such that for small $E$ and up to lower order terms 
\begin{equation}
\Nn(E)-\Nn(0)\,\sim\,\frac{C}{\log(E)^2}
\;,
\qquad
\gamma(E)\,\sim\,\frac{C'}{|\log(E)|}
\;.
\label{eq-Asymptotics}
\end{equation}
These expansions can be found in the physics literature \cite{EM,HJ,TC} on the random hopping model, but rigorous proofs were only given more recently in \cite{KV,DSS, DSS2}. Hence the localization length (aka inverse Lyapunov exponent) diverges slowly for states near the critical energy $E_c=0$, and there are typically many such states due to the Dyson peak. In another random quantum model, the so-called random dimer model, a vanishing of the Lyapunov exponent was argued to lead to quantum diffusion \cite{DWP}. This was rigorously confirmed in \cite{JSS,JS} where also the precise multi-scaling was determined. Yet another random model with a vanishing Lyapunov exponent is a particular type of random Kronig-Penney model which also exhibits quantum diffusion \cite{DKS}. Compared to \eqref{eq-Asymptotics}, these models behave considerably different because the Lyapunov exponents only grow with a power law $\gamma(E_c+\epsilon)\sim |\epsilon|^\alpha$ with $\alpha>0$ near the critical energy. This slower growth leads to more extended states and thus much faster quantum diffusion than the slow logarithmic growth associated with \eqref{eq-Asymptotics}. Clearly, the transport in all these models with a vanishing Lyapunov exponent at a critical energy is provided by the few states which are close to the critical energy, by a mechanism that is unrelated to the continuity of the spectral measures which can be exploited for quasiperiodic systems \cite{Gua}. If the Lyapunov exponent only vanishes at a finite number of points, it is possible to prove that the spectral measures of the infinite-volume Hamiltonian are almost surely pure-point. For the case of the random dimer model this is proved in \cite{BG} and for the critical Su-Schriefer-Heeger chain in \cite{Sha}, but we are not aware of a treatment of the random hopping model (even though the approach of \cite{BG} based on the positivity of the Lyapunov exponent away from $0$ possibly works; in particular, let us stress that the work \cite{Ran} allows for random hoppings parameters, but also requires the potentials to be random). Hence in this class of models, merely the extended nature of the eigenfunctions is responsible for the quantum transport. In fact, at (disorder-driven) Anderson transitions the eigenfunctions can be multifractal \cite{EM} and, as will be explained in this introduction further down, for the models studied here, the eigenfunctions of the low-lying eigenvalues do {\it not} have single localization centers which are Poisson distributed, which is a behavior known to hold in the Anderson-localized phase \cite{GK,Nak}. Finally, let us add that it has recently been shown that Sinai diffusion can also appear in almost periodic models \cite{JLM}.

\vspace{.2cm}

Let us now describe the main result in mathematical terms. The random hopping Hamiltonian $H^N$, with $N=2N'$ on $\phi=(\phi_n)_{n=1,\ldots,N}\in\ell^2(\{1,\ldots,N\})\cong \CM^N$ with Dirichlet boundary conditions is defined by
$$
(H^N\phi)_n
\;=\;
t_{n+1}\phi_{n+1}+t_{n}\phi_{n-1}
\;,
\qquad
\phi_0=\phi_{N+1}=0
\,,
$$ 
where $(t_{n})_{n=2,\ldots, N}$ are independent and identically distributed positive random variables drawn from a compact interval $[a,b]$ with $a>0$. Given some initial condition $\phi(0)\in \CM^N$, its quantum time-evolution is $\phi(t)=e^{-\imath H^N t}\phi(0)$. For technical convenience, the normalized initial state localized near the middle $N'$ of the sample will be chosen to be
\begin{equation}
\label{eq-InitialState}
\phi(0)
\;=\;
\frac{1}{\sqrt{2}}\big(|N'\rangle+|N'+1\rangle\big)
\;.
\end{equation}
Here $|n\rangle$ is the standard Dirac Ket notation for the normalized state at site $n\in\{1,\ldots,N\}$. Then the spreading of the state $\phi(t)$ will be studied via its time- and disorder-averaged $q$-th moments w.r.t. the shifted position operator defined by $X^N|n\rangle=(n-N')|n\rangle$, namely
$$
M^N_q(T)
\;=\;
\EM
\int_0^T \frac{dt}{T} \;
\| |X^N|^{\frac{q}{2}} \phi(t)\|^2
\;,
$$
where $\EM$ denotes the expectation over the random configurations. The main result gives a quantitative lower bound on $M^N_q(T)$, albeit only in a regime where time $T$ and length $N$ are roughly connected by $T\sim e^{\sqrt{N}}$, just as in \cite{BAK}.

\begin{theo}
\label{theo-Intro}
Let $\alpha>\frac{1}{2}$. There exist positive constants $c_2>c_1>2$ and $A_{q,\alpha}<\infty$ such that for $T$ satisfying $e^{c_1 N^\alpha}\leq T \leq e^{c_2 N^\alpha}$ one has
$$
M^N_q(T)
\;\geq\; A_{q,\alpha}\, \log(T)^{\frac{q-1}{\alpha}}\,-\,\mathcal{O}(N^{q+5}e^{-N^\alpha})
\;.
$$
\end{theo}

Let us follow up with a short discussion of the key elements of the proof of Theorem~\ref{theo-Intro}, which we also believe to unravel the physical mechanism of the quantum transport. First of all, one can express the moments $M^N_q(T)$ in a orthonormalized eigenbasis $(\phi^{E^N_n})_{n=1,\ldots,N}$ of $H^N$, having energies $(E^N_n)_{n=1,\ldots,N}$:
\begin{align}
M^N_q(T)
&
\;=\;
\EM
\sum_{n,m=1}^N
\langle \phi^{E^N_n}||X^N|^q|\phi^{E^N_m}\rangle\,\langle \phi(0)|\phi^{E^N_n}\rangle\,\langle \phi^{E^N_m}|\phi(0)\rangle
\,
\int^T_0\frac{dt}{T}\,e^{\imath (E^N_n-E_m^N)t}
\;.
\label{eq-heuristics}
\end{align}
The diagonal terms $n=m$ give positive contributions, while the off-diagonal terms $n\not=m$ are smaller provided that the time $T$ is sufficiently large compared to the inverse of the energy difference. The time scales in Theorem~\ref{theo-Intro} allow to show that these contributions are negligible so that one can focus on the diagonal ones. For those summands to provide a large contribution, it is necessary that both factors $\langle \phi^{E^N_n}||X^N|^q|\phi^{E^N_m}\rangle$ and $|\langle \phi(0)|\phi^{E^N_n}\rangle|^2$ are large. For states $\phi^{E^N_n}$ with one localization center, this is not possible because the second factor requires this localization center to be close to $N'$ which then, however, implies that the first factor is small. The crucial insight is now that the eigenstate with the smallest positive energy, denoted by $E^N_1$, has {\it two} localization centers which are, moreover, far apart, see Figure~\ref{fig-intro}. Theorem~\ref{theo-Intro} then follows by proving a lower bound on merely this one diagonal summand, combined with a rigorous control of the error terms. 

\vspace{.2cm}

The existence of two localization centers is rooted in the chiral (or sublattice) symmetry of the random hopping Hamiltonian, with a symmetry operator given by the parity of the position. In Section~\ref{sec-ChiralJac}, it will be shown that this implies that the spectrum is symmetric around $0$ and that the eigenfunctions $\phi^E$ for any eigenvalue $E$ satisfy
$$
\sum_{i=1}^{N'}(\phi^E_{2i-1})^2
\;=\;
\sum_{i=1}^{N'}(\phi^E_{2i})^2
\;,
$$
namely the odd and even components of $\phi^E=(\tilde{\phi}^E,\hat{\phi}^E)$ have the same weight. Moreover, both components are eigenstates of two naturally associated positive random Jacobi matrices which have localized eigenstates, albeit only with a subexponential decay from the localization center.  This is corroborated by a connection between the eigenvectors of the random hopping model and a centered random walk induced by the Hamiltonian is established which was already in the physics literature, {\it e.g.} \cite{ITA}. The random walk is defined by
\begin{equation}
\label{eq-RandWalkDef0}
\tilde{w}_1=0\;, \qquad 
\tilde{w}_{n+1}\,=\,\tilde{w}_{n}\,+\,\log(\tilde{\kappa}_n)
\;,
\qquad
\tilde{\kappa}_n
\,=\,
\frac{t_{2n}}{t_{2n+1}}\;, \qquad n=1,\ldots, N'-1.
\end{equation}
Note that the increments $\log(\tilde{\kappa}_n)$ are centered i.i.d. random variables. Under appropriate conditions and only on relevant parts of the sample, it is then shown that the components of the eigenvector of the smallest positive eigenvalue $E=E^N_1$ can be approximated by the exponential of this random walk:
\begin{equation}
\label{eq-IntroLink}
\tilde{\phi}^E\,\approx\,e^{\tilde{w}}
\;,
\qquad
\hat{\phi}^E\,\approx\,e^{\hat{w}}
\;,
\end{equation}
where $\hat{w}\approx-\tilde{w}$. Figure~\ref{fig-intro} shows the eigenfunction of the smallest eigenvalue and the random walk for a typical realization. The numerical agreement in \eqref{eq-IntroLink} is remarkably good (see also Figure~\ref{fig:wrapfig2} below). The control of the errors terms to these identities exploits algebraic features of chiral Jacobi matrices. Furthermore, new precise bounds on the smallest eigenvalue and the gap above it are proved using the random walks as trial states. These deterministic results formulated in detail in Section~\ref{sec-ChiralJac} make up the core analytic novelties of this work.

\vspace{.2cm}

Once the link \eqref{eq-IntroLink} between the eigenstate and the random walk is established, one can access the distribution of the two localization centers of $\phi^E$ by using classical results on random walks, such as the Sparre-Anderson theorem \cite{Fel} giving the distribution of the maximum of a random walk. Finer properties like the distribution of the separation of the distance between the maximum and minimum (aka the distance between the two localization centers) can be accessed using appropriate Donsker-type invariance theorems \cite{Bol,Igl} and explicit formulas for Brownian motion and Brownian meanders \cite{BS,MMS,SH}. The probabilistic statements which this implies on the eigenvalues and eigenstates of the random hopping Hamiltonian are given in Section~\ref{sec-Probabilistic}. These results are tailored for the proof of Theorem~\ref{theo-Intro} which is spelled out in Section~\ref{sec-RealEnergies} (this is partially based on the first {\tt arXiv} version of \cite{JSS}). The final Section~\ref{sec-NumIll} presents a few numerical results which illustrate the analytical findings, but also discusses the optimality of Theorem~\ref{theo-Intro} and raises some natural follow-up questions. 

\begin{figure}
    \centering
    \includegraphics[width=0.48\linewidth]{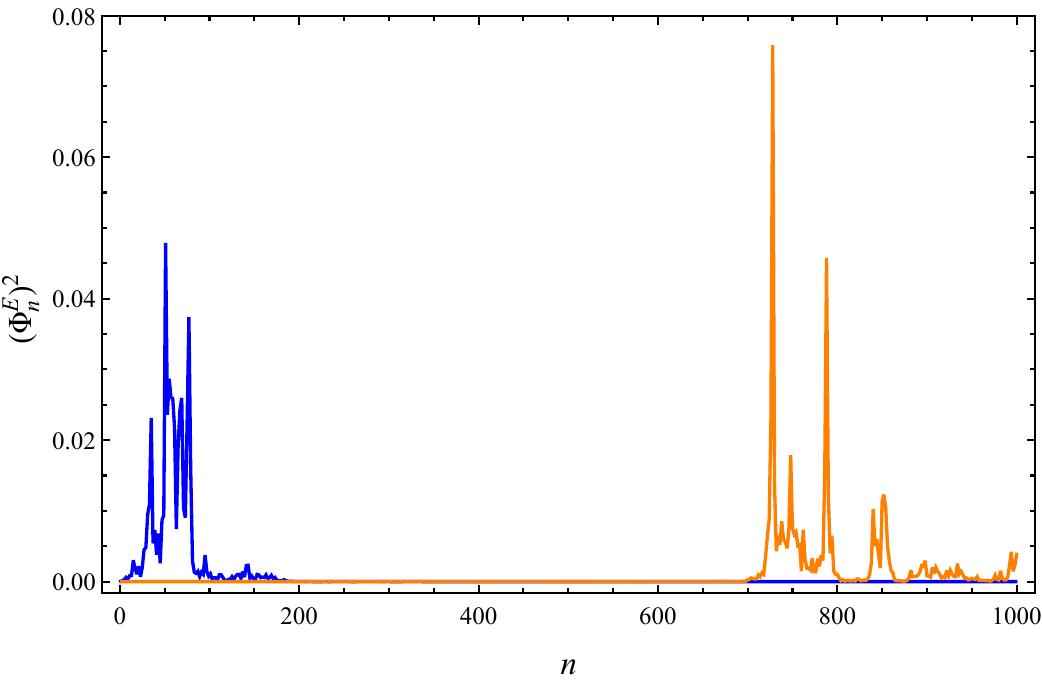}    
\hspace{.1cm}
    \includegraphics[width=0.48\linewidth]{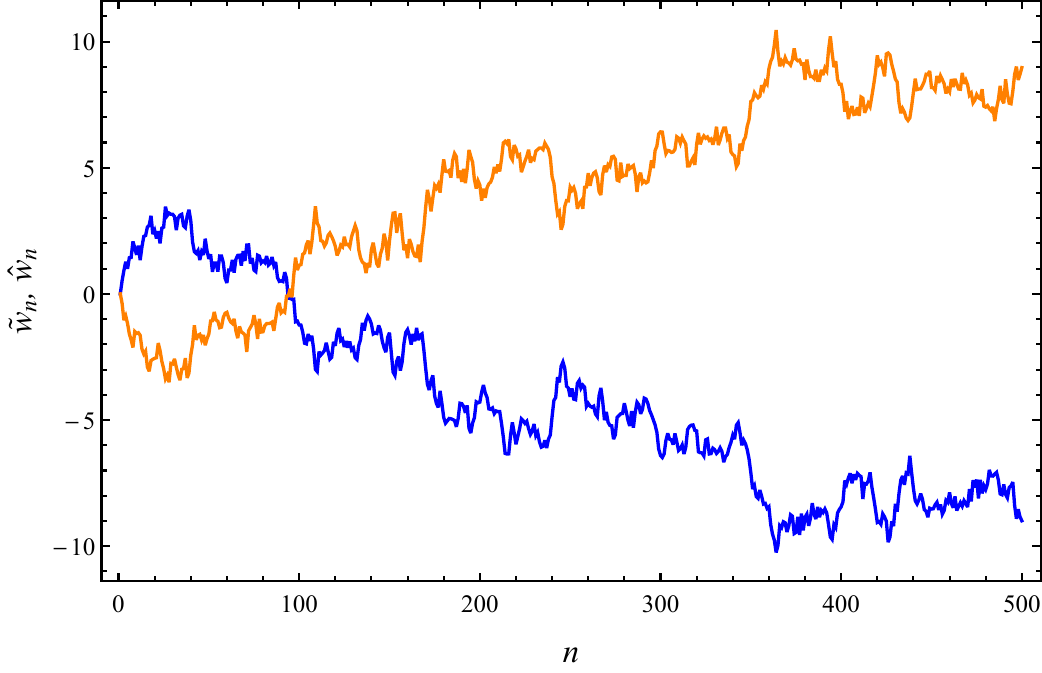}
\caption{\textit{The first plot shows  $n\in\{1,\ldots,N\}\mapsto |\phi^{E}_n|^2$ of $\phi^{E}\in\CM^N$ of the smallest positive eigenvalue $E=E^N_1$ for a typical realization $H^N$ of the random hopping Hamiltonian for $N =1000 $ and $t_n$ drawn uniformly from $[1.7,4.5]$. Here the two colors show the odd and even components of $\phi^{E}$ on the odd and even sites respectively. The right plot shows the associated random walks $\tilde{w}$ and $\hat{w}$ for the same configuration.
}}
\label{fig-intro}
\end{figure}

\vspace{.3cm}

\noindent {\bf Acknowledgements:} H.~S-B. thanks Matt Forster for several discussions on the project. This work was supported by the DFG grant SCHU 1358/8-1.

\section{Chiral Jacobi matrices}
\label{sec-ChiralJac}

The hopping Hamiltonian $H^N$ is a Jacobi matrix with vanishing diagonal
\begin{equation}
\label{eq-matrix}
H^N
\;=\;
\begin{pmatrix}
      & t_2  &        &        &         &        \\
t_2      &     &  t_3  &        &         &        \\
            & t_3 &     & \ddots &         &        \\
            &        & \ddots & \ddots & \ddots  &        \\
            &        &        & \ddots &  & t_N   \\
    &        &        &        & t_N  & 
\end{pmatrix}
\;,
\end{equation}
with the positive hopping coefficients $t_2,\ldots,t_N$ which are kept fixed in this section. The Hamiltonian possesses a crucial property which will be exploited throughout, namely, it has a {\it chiral symmetry} 
$$
J^NH^NJ^N
\;=\;
-\,H^N
\;,
$$
where $J^N$ is a selfadjoint which in position basis $(|n\rangle)_{1\leq n\leq N}$ is given by
$$
J^N|n\rangle
\;=\;
(-1)^{n-1}|n\rangle
\;.
$$
This section studies the general properties of the spectrum of $H^N$, its eigenvectors and especially focuses on the analysis of the eigenvector for the smallest positive eigenvalue. 

\subsection{Basic consequences of the chiral symmetry}

The spectrum of $H^N$ is simple and symmetric which means that $E \in \spec(H^N)$ implies $-E \in \spec(H^N)$. Indeed, denote the eigenvector for an eigenvalue $E$ by $\phi^E$, namely $H^N\phi^E=E\phi^E$, then the symmetry relation implies $-J^NH^NJ^N\phi^E=E\phi^E$ or $H^N(J^N\phi^E)=-E(J^N\phi^E)$. Hence $\phi^{-E}=J^N\phi^E$ and one deduces that 
$$
\phi^E_{2k-1}
\;=\;
\phi^{-E}_{2k-1}
\quad \text{and}\quad 
\phi^E_{2k}\;=\;
-\phi^{-E}_{2k}
\;, 
\qquad k=1,\dots,N' \,.
$$ 
This suggests to decompose the Hilbert space into odd and even sectors w.r.t. $J^N$, notably the Hilbert space is split $\Hh=\tilde{\Hh}\oplus \hat{\Hh}$ into a direct sum of the eigenspaces of $J^N$ for the eigenvalues $-1$ and $1$. Therefore any eigenvector of $H^N$ can be written as $\phi^E=\binom{\tilde{\phi}^E}{\hat{\phi}^E}$ with $\tilde{\phi}^E\in\tilde{\Hh}$ and $\hat{\phi}^E\in \hat{\Hh}$. Explicitly,
\begin{equation}
\label{eq-PhiTildeHat}
\tilde{\phi}^E_k\;=\;\phi^E_{2k-1}
\;,
\qquad
\hat{\phi}^E_k\;=\;\phi^E_{2k}
\;,
\qquad
k= 1,\ldots,N'
\;.
\end{equation}
By the above, $\phi^{-E}=\binom{\tilde{\phi}^E}{-\hat{\phi}^E}$. If now $E\not=0$, the orthogonality of $\phi^E$ and $\phi^{-E}$ implies $\|\tilde{\phi}^E\|^2=\|\hat{\phi}^E\|^2$. For the associated normalized state $\Phi^E=\|\phi^E\|^{-1}\phi^E$, one thus has  $\|\tilde{\Phi}^E\|^2=\frac{1}{2}=\|\hat{\Phi}^E\|^2$. Another important insight is obtained by writing $H^N$ in the grading of $J^N$:
\begin{equation}
\label{eq-ChiralA}
H^N
\;=\;
\begin{pmatrix} 0 & (A^N)^* \\ A^N & 0 \end{pmatrix}
\;,
\end{equation}
for some linear operator $A^N:\tilde{\Hh}\to\hat{\Hh}$.  Then 
$$
(H^N)^2
\;=\;
\begin{pmatrix}
(A^N)^*A^N & 0 \\ 0 & A^N(A^N)^*
\end{pmatrix}
\;.
$$
Concretely for the random hopping Hamiltonian, one finds
\begin{equation}
\label{eq-AevenIntro}
A^N
\;=\;
\begin{pmatrix}
t_2       & t_3  &        &        &         &        \\
   & t_4    &   t_5 &        &         &        \\
            &  & t_6    & \ddots &         &        \\
            &        &  & \ddots & \ddots  &        \\
            &        &        &  & t_{N-2} &t_{N-1}   \\
                 &        &        &        &   & t_{N}
\end{pmatrix}
\;,
\end{equation}
and
\begin{align}
\label{eq-A^*AIntro}
(A^N)^*A^N
&
\;=\;
\begin{pmatrix}
t_2^2       &  t_2t_3 &        &        &                \\
t_3t_2      & t_3^2+t_4^2    & t_4t_5   &                 &        \\
            & t_5t_4 & t_5^2+t_6^2   & \ddots &                 \\
            &        &  \ddots & \ddots &  t_{N-2}t_{N-1}         \\
     &        &        &         t_{N-1}t_{N-2}  & t_{N-1}^2+t_N^2
\end{pmatrix}
\;,
\end{align}
and similarly for $A^N(A^N)^*$. In particular, it is remarkable that $(A^N)^*A^N$, $A^N(A^N)^*$ and thus $(H^N)^2$ are again Jacobi matrices (hermitian tridiagonal) which by construction are positive definite. Inversely, given a positive Jacobi matrix, one can readily check that it can be written as such a product, a fact that will be exploited in a future work.  For the study of the eigenvalues of $H^N$, it is sufficient to carry out the spectral analysis of $(A^N)^*A^N$ because 
$$
(A^N)^*A^N\tilde{\phi}^E\;=\;E^2\tilde{\phi}^E
\quad\Longrightarrow\quad 
E\in \spec(H^N)
\;.
$$
Furthermore, if one also determines the eigenfunction $\hat{\phi}^E$ of $A^N(A^N)^*$ for the eigenvalue $E^2$,  one can readily construct the eigenstates $\phi^E$ and $\phi^{-E}$ of $H^N$. 

\subsection{Zero energy solutions as trial states}
\label{sec-ZeroEnergy}

Zero energy solutions $\psi$ (or so-called zero modes) satisfy $H^N\psi=0$. For even $N$, they cannot exist because all eigenvalues appear in pairs $(-E,E)$ and hence $0$ would have to be twice degenerate, which is impossible due to the simplicity of the spectrum of any Jacobi matrix (with periodic boundary conditions not considered here, a two-fold degenerate kernel is possible for so-called Brownian bridge configurations). On the other hand, by the same argument there must be exactly one zero energy solution for odd $N$ and this then has to lie in one of the sectors of $J^N$, more precisely the odd one due to the choice of $J^N$. In the following where $N=2N'$, zero energy solutions are solutions of the recurrence relation at energy $E=0$, but they do not satisfy the boundary conditions and are hence not eigenstates. Nevertheless, one expects them to be good approximations for eigenstates of energies which are close to zero, and this will be made explicit in Sections~\ref{sec-EigenfunctionEst} and \ref{sec-EigenfunctionEstNormalized}.

\vspace{.2cm}

By definition, the zero energy solution $\tilde{\psi}=(\tilde{\psi}_n)_{n=1,\ldots,N'}=(\phi^0_{2n-1})_{n=1,\ldots,N'}$ satisfies
$$
\tilde{\psi}_1\,=\,1
\;,
\qquad
t_{2n+1}\tilde{\psi}_{n+1}\,+\,t_{2n}\tilde{\psi}_n\,=\,0\;,
$$
 for $n=1,\ldots,N'-1$. Then one verifies
 \begin{equation}
\label{eq-Apsi}
 A^N{\tilde{\psi}}\;=\;
 \begin{pmatrix}
 0 \\
 \vdots \\
 0 \\
 t_{N}\tilde{\psi}_{N'}
 \end{pmatrix}
 \;.
 \end{equation}
Explicitly, the recurrence relation gives
\begin{equation}
\label{eq-PsiTilde}
\tilde{\psi}_n
\;=\;
(-1)\frac{t_{2n-2}}{t_{2n-1}}\,\tilde{\psi}_{n-1}
\;=\;
(-1)^{n-1}\prod^{n-1}_{k=1}\frac{t_{2k}}{t_{2k+1}}
\;,
\qquad
n=2,\ldots,N'
\;.
\end{equation}
In a similar way one can deduce the zero mode for the even part of the eigenvector:
$$
\hat{\psi}_{1}\,=\,1
\;,
\qquad
t_{2n}\hat{\psi}_n\,+\,t_{2n-1}\hat{\psi}_{n-1}\,=\,0\;,
$$
for $n=2,\dots, N'$. Then one sees that:
\begin{equation}
\label{eq-Astarpsi}
(A^N)^* \hat{\psi}
\;=\;
\begin{pmatrix}
t_2\hat{\psi}_1
\\
0
\\
\vdots
\\
0
\end{pmatrix}
\;=
\begin{pmatrix}
t_2
\\
0
\\
\vdots
\\
0
\end{pmatrix}
\end{equation}
The recurrence relation for the even sector gives
$$
\hat{\psi}_{n}\;=\;(-1)^{n-1}\prod_{k=1}^{n-1}\frac{t_{2k+1}}{t_{2k+2}}\;, \qquad
n=2,\ldots,N'
\;.
$$
Multiplying the expressions for $\tilde{\psi}_n$ and $\hat{\psi}_n$ directly leads to the relation
\begin{equation}
\label{eq-2RandomWalks}
\tilde{\psi}_n\hat{\psi}_n
\;=\;
\frac{t_2}{t_{2n}}\;, 
\qquad
n=1,\ldots,N'
\;.
\end{equation}
In the following sections, it will be shown that these zero energy solutions give, on certain parts of the system, good approximations of the eigenfunction of the smallest positive eigenvalue $E^N_1$ of $H^N$. As trial states, they will also provide good estimates on $E^N_1$.

\subsection{Reminder on oscillation theory}
\label{sec-Osci}

For any energy $E\in\RM$, the formal solution $\phi^E=(\phi^E_n)_{n=1,\ldots,N}$ of the eigenvalue equation $H^N\phi^E=E\phi^E$ with Dirichlet boundary condition at the left edge satisfies the three-term recurrence relation 
\begin{equation}
\label{eq-DirichletSolution}
t_{n+1}\phi^E_{n+1}+t_{n}\phi^E_{n-1}
\;=\;
E\,\phi^E_n
\;,
\qquad
\phi^E_0=0\;,
\;\;\;
\phi^E_{1}=1
\;,
\end{equation}
where $n=1,\ldots,N$. This solution is called formal because these relations produce some value $\phi^E_{N+1}\in\RM$ so that the right boundary condition $\phi^E_{N+1}=0$ may not be satisfied. Those energies $E\in\RM$ for which $\phi^E_{N+1}=0$ are precisely the eigenvalues of $H^N$. For $N$ odd, it was argued above that $\phi^0=\binom{\tilde{\psi}}{0}$ is an eigenstate. It has exactly $\frac{N-1}{2}$ oscillations, as can be read off from \eqref{eq-PsiTilde}. For $E>0$ but small, the zeros of $\phi^E$ at all even sites are lifted, and the signs of the components can be read off the recurrence relation, giving
$$
\sgn(\phi^E_{2k-1})
\;=\;
\sgn(\phi^E_{2k})
\;=\;
(-1)^{k-1}
\;.
$$
This holds as long as $E$ is less than or equal to the first positive eigenvalue $E^N_1>0$. According to \eqref{eq-PhiTildeHat} and \eqref{eq-PsiTilde} this also implies for all non-negative $E\leq E^N_1$
$$
\sgn(\phi^E_{2k-1})
\;=\;
\sgn(\tilde{\phi}^E_{k})
\;=\;
\sgn(\tilde{\psi}_{k})
\;=\;
(-1)^{k-1}
\;,
$$
and similarly
$$
\sgn(\phi^E_{2k})
\;=\;
\sgn(\hat{\phi}^E_{k})
\;=\;
\sgn(\hat{\psi}_{k})
\;=\;
(-1)^{k-1}
\;.
$$
This reflects that $E=0$ lies in the middle of the spectrum and therefore there are $\frac{N}{2}$ space oscillations of the solution $\phi^0$ (and solutions nearby). It is a special property of a chiral Hamiltonian that these oscillations are $2$-periodic.

\subsection{Control of eigenstate for smallest eigenvalue}
\label{sec-EigenfunctionEst}

The state $\psi=\binom{\tilde{\psi}}{0}$ and the eigenstate $\phi^{E^N_1}$ both satisfy the same left Dirichlet boundary condition and they almost satisfy the same three-term recurrence relation, the difference only resulting from the (small) difference in energy between $0$ and $E^N_1>0$. The following statement controls the quotient of these states for all energies $E\leq E^N_1$ in an efficient quantitative manner. In fact the result implies that this quotient is close to $1$ on the left side of the sample $\{1,\ldots,N\}$, but necessarily deviates towards the right side. Similarly, the state $\psi=\binom{0}{\hat{\psi}}$ and $\phi^{E^N_1}$ satisfy after suitable normalization the same right Dirichlet boundary condition and almost satisfy the same three-term recurrence relation, so for sites close to the right boundary they are expected to be almost the same. This is shown in the second claim of the next proposition.

\begin{proposi}
\label{prop-EigEst} 
For all positive $E\leq E^N_1$ and even $N=2N'$, the functions
$$
k\in\{1,\ldots,N'\}\;\mapsto\;\frac{\tilde{\phi}^{E}_{k}}{\tilde{\psi}_k}
\quad
\mbox{ and }
\quad
k\in\{1,\ldots,N'\}\;\mapsto\;\frac{\hat{\phi}^{E}_{N'-k}}{\hat{\psi}_{N'-k}}
$$
are decreasing and satisfy the bounds
\begin{equation}
\label{eq-LowestStateBound1}
1-\left(\frac{E}{t_2}\right)^{\!2}
\;\sum_{p=1}^{k-1}\hat{\psi}_{p}^2\sum_{l=1}^p\tilde{\psi}_{l}^2
\;\leq \;
\frac{\tilde{\phi}^{E}_{k}}{\tilde{\psi}_k}
\;\leq\;
1
\;,
\end{equation}
as well as
\begin{equation}
\label{eq-LowestStateBound2}
1\,-\,\left(\frac{E}{t_2}\right)^{\! 2}
\sum_{p=0}^{k-1}\tilde{\psi}_{N'-p}^2\sum_{l=0}^p\hat{\psi}_{N'-l}^2
\;\leq\;
\frac{\hat{\psi}_{N'}}{\hat{\phi}^{E}_{N'}}
\frac{\hat{\phi}^{E}_{N'-k}}{\hat{\psi}_{N'-k}}
\;\leq\; 
1\;.
\end{equation}
\end{proposi}

\begin{rem}
{\rm
As the energy $E^N_1$ will be shown to be very small of the order $e^{-\sqrt{N}}$, the lower bound in \eqref{eq-LowestStateBound1} for $E=E^N_1$ is close to $1$ for small $k$ , showing that the eigenfunction $\phi^{E^N_1}$ is very well approximated by the zero energy state $\binom{\tilde{\psi}}{0}$ close to the left edge (sufficiently small $k$). However, the quality of the lower bound  in \eqref{eq-LowestStateBound1} decreases as $k$ increases, and, in particular, it is of no interest at the right boundary. On the other hand, the lower bound in \eqref{eq-LowestStateBound2} shows that $\phi^{E^N_1}$ is well approximated by the zero energy state $\binom{0}{\hat{\psi}}$ close to the right  edge (again $k$ small), up to the suitable normalization factor $\frac{\hat{\psi}_{N'}}{\hat{\phi}^{E}_{N'}}$ assuring that one starts with Dirichlet boundary condition at the right edge $N'$. Again the lower bound in \eqref{eq-LowestStateBound2} is of little value at the left boundary (for $k$ close to $N'$). Somewhere in the middle of the sample is a transition from $\binom{\tilde{\psi}}{0}$ to $\binom{0}{\hat{\psi}}$ resulting from one oscillation. Typically this transition is in a small spatial region and is linked to bursts in a certain random dynamical system. This will be analyzed in more detail elsewhere.  
}
\hfill $\diamond$
\end{rem}

\begin{rem}
{\rm
The proof of \eqref{eq-LowestStateBound1} can be transposed to the Hamiltonian restricted to $\{2,\ldots,N\}$ for which the Dirichlet boundary condition corresponds to the anti-Dirichlet boundary condition of $H^N$. More precisely, if $\eta^E=(\eta^E_n)_{n=2,\ldots,N}$ is defined as the solution of the three-term recurrence relation with initial conditions $\eta^E_1=0$ and $\eta^E_2=1$, then  for $k\geq 2$
\begin{equation}
\label{eq-LowestStateBound3}
1-\left(\frac{E}{t_3}\right)^{\!2}
\;\sum_{p=2}^{k-1}\tilde{\psi}_{p}^2\sum_{l=2}^p\hat{\psi}_{l}^2
\;\leq \;
\frac{\tilde{\eta}^{E}_{k}}{\hat{\psi}_k}
\;\leq\;
1
\;.
\end{equation}
While similar to the bound \eqref{eq-LowestStateBound2}, the latter has the advantage to involve the solution $\phi^E$ with Dirichlet boundary condition on the left edge, and hence the eigenstate $\phi^{E^N_1}$.
}
\hfill $\diamond$
\end{rem}

The proof of Proposition~\ref{prop-EigEst} and many other statements below will be using the standard transfer matrices at energy $E\in\RM$ . Recall ({\it e.g.} \cite{JSS}) that for integer sites $n\geq m\geq 1$ they are defined by
$$
\Tt^E_n
\;= \;
\begin{pmatrix}
E \,\frac{1}{t_n} & -t_n \\
\frac{1}{t_n} & 0
\end{pmatrix}
\;,
\qquad
\mathcal{T}^E(n,m) \;=\; \Tt^E_{n-1}  \cdots \Tt^E_m
\;,
$$
where $\mathcal{T}^E(m,m)=\one$ and $t_1=1$. For $m>n$, one also sets $\Tt^E(n,m)=\Tt^E(m,n)^{-1}$. The transfer matrices allow to rewrite the solution $\phi^E$ of the three-term recurrence relation \eqref{eq-DirichletSolution} as
\begin{equation}
\label{eq-transferrelOld}
\left(\begin{array}{c} t_{n+1}\,\phi^E_{n+1} \\ \phi^E_{n}
\end{array} \right)
\;=\;
{\cal T}^E(n+1,m)\;
\left(\begin{array}{c}\, t_m\,\phi^E_m \\ \phi^E_{m-1}
\end{array} \right)
\;.
\end{equation}
For another energy $\epsilon\in\RM$, each matrix can be decomposed as
$$
\Tt^E_n
\;=\; 
\Tt^\epsilon_n\,+\, (E-\epsilon)\frac{1}{t_n}
\begin{pmatrix}
1 & 0 \\
0 & 0
\end{pmatrix}
\;.
$$
Applying this iteratively leads to the well-known variation of parameters formula
\begin{equation}
\label{eq-transferid3}
{\cal T}^{E}(n+1,m)
\;=\;
{\cal T}^\epsilon(n+1,m)
+(E-\epsilon)\,
\sum_{l=m}^{n}\,
{\cal T}^{\epsilon}(n,l+1)\;
\frac{1}{t_l}\,
\left(
\begin{array}{cc} 1 & 0 \\
0 & 0 
\end{array}
\right)\,
{\cal T}^E(l,m)
\;.
\end{equation}
Replacing this for $m=1$ in \eqref{eq-transferrelOld} and extracting the lower component leads to
\begin{align}
\phi^E_n
&
\;=\;
\binom{0}{1}^*
{\cal T}^\epsilon(n+1,1)
\binom{1}{0}
\nonumber
\\
&
\;=\;
\phi^\epsilon_n+(E-\epsilon)\sum_{l=1}^n \binom{0}{1}^*
{\cal T}^\epsilon(n+1,l+1)
\binom{1}{0}\frac{1}{t_{l+1}}\phi^E_l
\;.
\label{TM-Sol}
\end{align}
where $\phi^E, \phi^\epsilon$ satisfy the same initial condition, but are not normalized. Of particular importance will be the transfer matrices at zero energy. They can readily be computed explicitly:

\begin{lemma}
\label{lem-T0}
If $n$ and $l$ are of same parity, the matrix ${\cal T}^0(n,l)$ is diagonal. If $n$ and $l$ are of different parity, ${\cal T}^0(n,l)$ is off-diagonal and its matrix elements are given by
\begin{equation}
\label{eq-T0matrix}
\binom{0}{1}^*
\mathcal{T}^{0}(2k+1,2p)\begin{pmatrix}1\\0\end{pmatrix}
\;=\;
\frac{\hat{\psi}_k}{\hat{\psi_p}}
\;,
\qquad
\binom{0}{1}^*
\mathcal{T}^{0}(2k,2p+1)\begin{pmatrix}1\\0\end{pmatrix}
\;=\;
\frac{\tilde{\psi}_{k}}{\tilde{\psi}_{p+1}}
\end{equation}
\end{lemma}

\noindent {\bf Proof} of Proposition~\ref{prop-EigEst}.
Choosing $\epsilon=0$ in \eqref{TM-Sol} gives
$$
\phi^{E}_n 
\;=\; 
\phi^{0}_n \,+\, E\sum_{l=1}^n 
\binom{0}{1}^*
\mathcal{T}^0(n+1,l+1)\begin{pmatrix}1\\0\end{pmatrix}
\frac{1}{t_{l+1}}\,\phi^{E}_l 
\;.
$$
Focussing first on odd $n=2k-1$ and using the first claim of Lemma~\ref{lem-T0}, this becomes
\begin{align*}
\tilde{\phi}^{E}_{k} 
& 
\;=\; 
\tilde\psi_k \,+\, E\sum_{l=1}^{2k-1} \binom{0}{1}^*
\Tt^{0}(2k,l+1)\begin{pmatrix}1\\0\end{pmatrix}
\frac{1}{t_{l+1}}\,\phi^{E}_l 
\\
&
\;=\; 
\tilde\psi_k \,+\, E\sum_{p=1}^{k-1} \binom{0}{1}^*
\mathcal{T}^{0}(2k,2p+1)\begin{pmatrix}1\\0\end{pmatrix}
\frac{1}{t_{2p+1}}\,\phi^{E}_{2p}
\\
&
\;=\;\tilde\psi_k \,+\, E\sum_{p=1}^{k-1} \binom{0}{1}^*
\mathcal{T}^{0}(2k,2p+1)\begin{pmatrix}1\\0\end{pmatrix}
\frac{1}{t_{2p+1}}\,\hat{\phi}^{E}_{p} 
\;. 
\end{align*}
Thus the second part of Lemma~\ref{lem-T0} implies
$$
\tilde{\phi}^{E}_{k} 
\;=\; 
\tilde\psi_k + E\sum_{p=1}^{k-1}\frac{\tilde{\psi}_{k}}{\tilde{\psi}_{p+1}}
\frac{1}{t_{2p+1}}\,\hat{\phi}^{E}_{p}
\;=\;
\tilde\psi_k - E\sum_{p=1}^{k-1}\frac{\tilde{\psi}_{k}}{t_{2p}\tilde{\psi}_{p}}
\,\hat{\phi}^{E}_{p}
\;,
$$
where in the second equality the recurrence relation was used. Taking the ratios gives
\begin{equation}
\label{eq-Bound0}
\frac{\tilde{\phi}^{E}_{k}}{\tilde{\psi}_k}
\;=\;
1-E\sum_{p=1}^{k-1}\frac{\hat{\psi}_{p}}{t_{2p}\tilde{\psi}_{p}}
\,\frac{\hat{\phi}^{E}_{p}}{\hat{\psi}_p}
\;=\;
1-\frac{E}{t_2}
\sum_{p=1}^{k-1}\hat{\psi}_{p}^2
\,\frac{\hat{\phi}^{E}_{p}}{\hat{\psi}_p}
\;,
\end{equation}
where \eqref{eq-2RandomWalks} was used. By Sturm-Liouville oscillation theory described in Section~\ref{sec-Osci}, $\hat{\phi}^{E}_{p}$ and $\hat{\psi}_p$ have the same (alternating) sign for all $p$, because $E$ is less than or equal to the smallest positive eigenvalue $E^N_1$ by assumption. Hence all summands on the r.h.s. are positive and one concludes that
\begin{equation}
\label{eq-Bound1}
\frac{\tilde{\phi}^{E}_{k}}{\tilde{\psi}_k}
\;\leq\;1
\;.
\end{equation}
In order to provide a lower bound, let us consider the case of even $n=2k$ and $m=1$ in \eqref{eq-transferid3} with $\epsilon=0$. Then extracting the lower left corner one finds
\begin{align*}
\hat{\phi}^{E}_{k} 
&
\;=\; E\sum_{l=1}^{2k} \binom{0}{1}^*
\mathcal{T}^{0}(2k+1,l+1)\begin{pmatrix}1\\0\end{pmatrix}
\frac{1}{t_{l+1}}\,\phi^{E}_l 
\\
&
\;=\;
E\sum_{p=1}^{k} \binom{0}{1}^*
\mathcal{T}^{0}(2k+1,2p)\begin{pmatrix}1\\0\end{pmatrix}
\frac{1}{t_{2p}}\,\phi^{E}_{2p-1}
\\
&
\;=\;
E\sum_{p=1}^{k} \binom{0}{1}^*
\mathcal{T}^{0}(2k+1,2p)\begin{pmatrix}1\\0\end{pmatrix}
\frac{1}{t_{2p}}\,\tilde{\phi}^{E}_{p}
\\
&
\;=\;
E\sum_{p=1}^{k} 
\frac{\hat{\psi}_k}{\hat{\psi_p}t_{2p}}
\,\tilde{\phi}^{E}_{p}
\;.
\end{align*}
Hence
\begin{equation}
\label{eq-QuotientIncrease}
\frac{\hat{\phi}^{E}_{k}}{\hat{\psi}_k}
\;=\;
E\sum_{p=1}^{k}\frac{\tilde{\psi}_{p}}{t_{2p}\hat{\psi}_{p}}
\,\frac{\tilde{\phi}^{E}_{p}}{\tilde{\psi}_p}
\;=\;
\frac{E}{t_2}\sum_{p=1}^{k}\tilde{\psi}_{p}^2 \,\frac{\tilde{\phi}^{E}_{p}}{\tilde{\psi}_p}
\;\leq\; 
\frac{E}{t_2}\sum_{p=1}^{k}\tilde{\psi}_{p}^2
\;,
\end{equation}
because of  \eqref{eq-Bound1}. Replacing this into  \eqref{eq-Bound0} gives the lower bound in \eqref{eq-LowestStateBound1}.

\vspace{.1cm}

For the upper bound in \eqref{eq-LowestStateBound2}, let us note that the quotient on the l.h.s. of \eqref{eq-QuotientIncrease} is increasing in $k$ because all summands after the first equality are non-negative. Hence the largest value is taken at $k=N'$. Dividing by this largest value immediately implies the upper bound in \eqref{eq-LowestStateBound2}. For the lower bound, one can mirror the Hamiltonian as well as the states involved via
$$
t_n \,\mapsto\, t_{N+2-n}\;,
\qquad
\phi^E_n\,\mapsto\,\frac{\phi^E_{N+1-n}}{\phi^E_N}
\;,
\qquad
\hat{\psi}_p \,\mapsto\, 
\frac{\tilde{\psi}_{N'-p+1}}{\tilde{\psi}_{N'}}\;, 
\qquad 
\tilde{\psi}_l\,\mapsto\,  \frac{\hat{\psi}_{N'-l+1}}{\hat{\psi}_{N'}} 
\;. 
$$
The normalization factors assure that the reflected states satisfy the left boundary condition (at $n=0$ and $n=1$). Therefore the lower bound in \eqref{eq-LowestStateBound1} implies
$$
\frac{\hat{\psi}_{N'}}{\hat{\psi}_{N'-k}}\,
\frac{\hat{\phi}^{E}_{N'-k}}{\hat{\phi}^{E}_{N'}}
\; \geq \; 
1\,-\,
\left(\frac{E}{t_N}\right)^{\!2} \sum_{p=1}^{k}\left(\frac{\tilde{\psi}_{N'-p+1}}{\tilde{\psi}_{N'}}\right)^2\sum_{l=1}^p\left(\frac{\hat{\psi}_{N'-l+1}}{\hat{\psi}_{N'}}\right)^2 
\;.
$$
Due to the identities $\phi^E_N=\hat{\phi}^E_{2N'}$ and $t_N\tilde{\psi}_{N'}\hat{\psi}_{N'}=t_2$ this is the lower bound in \eqref{eq-LowestStateBound2}.
\hfill $\Box$

\subsection{Control of normalized eigenstate for smallest eigenvalue}
\label{sec-EigenfunctionEstNormalized}

In this short section, Proposition~\ref{prop-EigEst} will be extended to obtain information on the normalized eigenstate of the smallest eigenvalue. Let introduce the notations 
$$
\Phi^{E^N_1}
\;=\;
\frac{1}{\|\phi^{E^N_1}\|}\,\phi^{E^N_1}
\;=\;
\frac{1}{\sqrt{2}}\left(\frac{\tilde{\phi}^{E^N_1}}{\|\tilde{\phi}^{E^N_1}\|},\frac{\hat{\phi}^{E^N_1}}{\|\hat{\phi}^{E^N_1}\|}\right)^T
\;=\;
(\tilde{\Phi}^{E^N_1},\hat{\Phi}^{E^N_1})^T
\;.
$$
Then $\|\tilde{\Phi}^{E^N_1}\|^2=\|\hat{\Phi}^{E^N_1}\|^2=\frac{1}{2}$. Similarly set 
\begin{equation}
\label{eq-PsiNormalized}
\tilde{\Psi}
\;=\;
\frac{1}{\sqrt{2}\|\tilde{\psi}\|}\,\tilde{\psi}
\;,
\qquad
\hat{\Psi}
\;=\;
\frac{1}{\sqrt{2}\|\hat{\psi}\|}\,\hat{\psi}
\;.
\end{equation}
Let us first bound the normalized states using Proposition~\ref{prop-EigEst}, namely for all $k\leq N'$ and $E\leq E^N_1$, one has due to $\|\tilde{\phi}^E\|^2\leq \|\tilde{\psi}\|^2$:
\begin{align}
|\tilde{\Phi}_k^E|
&
\;\geq \;
\frac{|\tilde{\phi}^E_k|}{\sqrt{2}\,\|\tilde{\psi}\|}
\nonumber
\\
&
\;\geq\; 
\frac{|\tilde{\psi}_k|}{\sqrt{2}\,\|\tilde{\psi}\|}\left(1-\left(\frac{E}{t_2}\right)^{\!2}
\;\sum_{p=1}^{k-1}\hat{\psi}_{p}^2\sum_{l=1}^p\tilde{\psi}_{l}^2\right)
\nonumber
\\
&
\;=\;
\tilde{|\Psi}_k|\left(1-\left(\frac{E}{t_2}\right)^{\!2}
\;\sum_{p=1}^{k-1}\hat{\psi}_{p}^2\sum_{l=1}^p\tilde{\psi}_{l}^2\right)
\;.
\label{eq-PhipsiBound0}
\end{align}
Another lower bound on $\tilde{\Phi}_k^E$ by $\tilde{\Psi}_k$ which does not involve correction terms can be obtained when spatial averages are taken:

\begin{proposi}
\label{prop-PhiPsiBound}
One has for any $n\leq N'$ and $E\leq E^1_N$
$$
\sum_{k=1}^n
(\tilde{\Phi}^{E}_k)^2
\;\geq\;
\sum_{k=1}^n
(\tilde{\Psi}_k)^2
\;.
$$
\end{proposi}

\noindent {\bf Proof.} Let us set
$$
\tilde{\alpha}_k\;=\;\frac{\tilde{\phi}^{E}_{k}}{\tilde{\psi}_k}
\;.
$$
By Proposition~\ref{prop-EigEst}, the function $k\in\{1,\ldots,N'\}\mapsto\tilde{\alpha}_k$ is non-increasing. The claim is equivalent to
$$
\sum_{k=1}^n
(\tilde{\phi}^{E}_k)^2
\sum_{l=1}^{N'}
(\tilde{\psi}_l)^2
\;\geq\;
\sum_{k=1}^n
(\tilde{\psi}_k)^2
\sum_{l=1}^{N'}
(\tilde{\phi}^{E}_l)^2
\;,
$$
or equivalently
$$
\sum_{k=1}^n
\tilde{\alpha}_k^2(\tilde{\psi}_k)^2
\sum_{l=1}^{N'}
(\tilde{\psi}_l)^2
\;\geq\;
\sum_{k=1}^n
(\tilde{\psi}_k)^2
\sum_{l=1}^{N'}
\tilde{\alpha}_l^2(\tilde{\psi}_l)^2
\;.
$$
All terms with $k,l\leq n$ cancel out, so that the statement is equivalent to 
$$
\sum_{k=1}^n
\sum_{l=n+1}^{N'}
(\tilde{\psi}_k)^2
(\tilde{\psi}_l)^2
\big(\tilde{\alpha}_l^2-\tilde{\alpha}_k^2\big)
\;\geq\;
0
\;.
$$
This holds because $l>k$.
\hfill $\Box$

\vspace{.2cm}

For the even part, let us first note that the right bound in \eqref{eq-LowestStateBound2} gives
$$
\|\hat{\phi}^{E}\|
\;\leq\;
\frac{\hat{\phi}^{E}_{N'}}{\hat{\psi}_{N'}}
\,
\|\hat{\psi}^{E}\|
\;,
$$
and the first bound in \eqref{eq-LowestStateBound2} then implies by the same argument leading to \eqref{eq-PhipsiBound0}
\begin{equation}
\label{eq-PhipsiBound}
|\hat{\Phi}^E_{N'-k}|
\;\geq\; 
|\hat{\Psi}_{N'-k}|\left(1\,-\,
\left(\frac{E}{t_2}\right)^{\!2} \sum_{p=1}^{k}\tilde{\psi}_{N'-p+1}^2\sum_{l=1}^p\hat{\psi}_{N'-l+1}^2\right)
\;.
\end{equation}
%

\subsection{Bounds on the smallest energy}
\label{sec-E1Bounds}

\begin{proposi}
\label{prop-MinMax}
The following bounds hold:
\begin{equation}
\label{eq-MinMax}
(E^N_1)^2 
\;\leq\; 
\frac{t_N^2\tilde{\psi}_{N'}^2}{\langle\tilde{\psi}|\tilde{\psi}\rangle}
\;,
\qquad
(E^N_1)^2
\;\leq\;
\frac{t_2^2}{\langle\hat{\psi}|\hat{\psi}\rangle}
\;.
\end{equation}
\end{proposi}

\noindent {\bf Proof.} The min-max principle implies
$$
(E^N_1)^2 
\;\leq\; 
\frac{\langle \tilde{\psi}|(A^N)^*A^N|\tilde{\psi}\rangle}{\langle\tilde{\psi}|\tilde{\psi}\rangle}
\;=\;
\frac{t_N^2\tilde{\psi}_{N'}^2}{\langle\tilde{\psi}|\tilde{\psi}\rangle}\; ,
\qquad
(E^N_1)^2 
\;\leq\; 
\frac{\langle \hat{\psi}|A^N(A^N)^*|\hat{\psi}\rangle}{\langle\hat{\psi}|\hat{\psi}\rangle}
\;=\;
\frac{t_2^2}{\langle\hat{\psi}|\hat{\psi}\rangle}\; ,
$$
where the equalities follow from (\ref{eq-Apsi}) and (\ref{eq-Astarpsi}).
\hfill $\Box$

\vspace{.2cm}

Let us note that, while these bounds look somewhat different, they are essentially the same because the first one results from the second one when one starts from the right side of the sample with boundary conditions $\tilde{\psi}_{N'}=1$. This does not mean, however, that both bounds are equally tight. In fact, it will be discussed in Remark~\ref{rem-OBC1} following up on Proposition~\ref{prop-EnergyBoundOBC} which of the bounds is more effective for a given realization. Both of the upper bounds on $E^N_1$ given in Proposition~\ref{prop-MinMax} can be combined with the lower bounds in \eqref{eq-PhipsiBound0} and \eqref{eq-PhipsiBound} (here the focus is on the second bound in \eqref{eq-MinMax}; moreover, combining with \eqref{eq-LowestStateBound1} and \eqref{eq-LowestStateBound2} is possible as well, but not spelled out). This immediately implies the following:

\begin{coro}
\label{coro-QuotientRandomWalkBound}
One has for $E\leq E^N_1$
$$
\frac{\tilde{\Phi}^{E}_{k}}{\tilde{\Psi}_k}
\;\geq \;
1\,-\,
\frac{1}{\langle\hat{\psi}|\hat{\psi}\rangle}
\;\sum_{p=1}^{k-1}\hat{\psi}_{p}^2\sum_{l=1}^p\tilde{\psi}_{l}^2
\;.
$$
as well as
$$
\frac{\hat{\Phi}^E_{N'-k}}{\hat{\Psi}_{N'-k}}
\;\geq\; 
1\,-\,
\frac{1}{\langle\hat{\psi}|\hat{\psi}\rangle}
\sum_{p=1}^{k}\tilde{\psi}_{N'-p+1}^2\sum_{l=1}^p\hat{\psi}_{N'-l+1}^2
\;.
$$
\end{coro}

The next aim is to provide a (still deterministic) lower bound on the smallest positive eigenvalue $E^N_1$ of a chiral Jacobi matrix $H^N$. 

\begin{proposi}
\label{prop-LowerE1}
One has
\begin{equation}
\label{eq-LowerE1}
E^N_1
\;\geq\; 
\frac{a}{2\,N}
\,
\min_{1\leq l\leq n\leq N'}
\Big\{\Big|\frac{\hat{\psi}_l}{\hat{\psi}_n}\Big|,\Big|\frac{\tilde{\psi}_l}{\tilde{\psi}_n}\Big|\Big\}
\;.
\end{equation}
\end{proposi}

The proof of this statement is based on a variation of constants argument contained in the work of Last and 
Simon \cite{LaSi}. It holds for arbitrary Jacobi matrices, namely does not require the chiral symmetry.

\begin{lemma}[based on a modification of the proof of Theorem~2.3 in \cite{LaSi}] 
\label{lem-LastSimon}
For $\epsilon\in\RM$, introduce a matrix $K^\epsilon_N\in\CM^{N\times N}$ by setting
\begin{equation}
\langle n|K^\epsilon_N|l\rangle
\;=\;
\delta_{l\leq n}\;
\frac{1}{t_{l+1}}\binom{0}{1}^*
{\cal T}^{\epsilon}(n+1,l+1)\,
\binom{1}{0}
\;.
\label{eq-KOpDef2}
\end{equation}
Let $E_n^N$ and $E_{n+1}^N$ be two adjacent eigenvalues of a finite Jacobi matrix $H^N$. Then for any $\epsilon\in [E_n^N,E_{n+1}^N]$,
\begin{equation}
|E^N_{n+1}-E^N_n|
\;\geq\;
\frac{1}{\|K^\epsilon_N\|}
\;,
\label{eq-LowerE1}
\end{equation}
where $\|K^\epsilon_N\|$ denotes the operator norm of $K^\epsilon_N$.
\end{lemma}

\noindent {\bf Proof:} Set $E_1=E^N_n$ and $E_2=E^N_{n+1}$ and let two corresponding eigenvectors of $H^N$ be  $\phi^1=\phi^{E^N_n}$ and $\phi^2=\phi^{E^N_{n+1}}$, both calculated with the recurrence relations at $E_1$ and $E_2$ respectively, both with Dirichlet boundary conditions as in \eqref{eq-DirichletSolution}. The right boundary condition then gives $\phi^1(N+1)=0=\phi^2(N+1)$. These states are orthogonal in $\CM^N$, namely
$$
\langle \phi^1|\phi^2\rangle_{N}\;=\;0
\;.
$$
Using the notation \eqref{eq-KOpDef2}, the identity \eqref{TM-Sol} becomes
\begin{equation}
\phi^E
\;=\;
\phi^\epsilon
\,+\,
(E-\epsilon)\,K^\epsilon_N\phi^E
\;.
\label{eq-PertEps3}
\end{equation}
Now consider \eqref{eq-PertEps3} for $E=E_1$ and then take the scalar product with $\phi^1$. One obtains
$$
\|\phi^1\|^2
\;=\;
\langle \phi^1|\phi^\epsilon\rangle_{N}
\,+\,
(E_1-\epsilon)\,
\langle \phi^1|K^\epsilon_N|\phi^1\rangle_{N}
\;.
$$
Similarly, taking \eqref{eq-PertEps3} for $E=E_2$ and again the scalar product with $\phi^1$,
$$
\langle \phi^1|\phi^2\rangle_{N}
\;=\;
\langle \phi^1|\phi^\epsilon\rangle_{N}
\,+\,
(E_2-\epsilon)\,
\langle \phi^1|K^\epsilon_N|\phi^2\rangle_{N}
\;,
$$
of which the l.h.s. vanishes. The difference between these equations is
\begin{equation}
\label{eq-Phi1Idbis}
\|\phi^1\|^2
\;=\;
(E_1-\epsilon)\,
\langle \phi^1|K^\epsilon_N|\phi^1\rangle_{N}
\,-\,
(E_2-\epsilon)\,
\langle \phi^1|K^\epsilon_N|\phi^2\rangle_{N}
\;.
\end{equation}
Since the eigenfunctions generated in this way are not normalized, one may assume, without loss of generality:
$$
\|\phi^1\|\;\geq\;\|\phi^2\|
\;.
$$
Then
$$
\|\phi^1\|^2
\;\leq\;
|E_1-\epsilon|\,\|\phi^1\|^2\,\|K^\epsilon_N\|
\,+\,
|E_2-\epsilon|\,\|\phi^1\|\,\|\phi^2\|\,\|K^\epsilon_N\|
\;\leq\;
(|E_1-\epsilon|+|E_2-\epsilon|)\,\|\phi^1\|^2\,\|K^\epsilon_N\|
\;,
$$
which directly implies the statement.
\hfill $\Box$

\vspace{.2cm}

\noindent {\bf Proof} of Proposition~\ref{prop-LowerE1}. Let us apply Lemma~\ref{lem-LastSimon} to $E_{n}^N=-E^N_1$ and $E_{n+1}^N=E^N_1$. Choosing $\epsilon=0\in[-E^N_1,E^N_1]$ one deduces
\begin{equation}
\label{eq-E1K}
2E^N_1
\;\geq\; 
\frac{1}{\|K^0_N\|} 
\;.
\end{equation}
Hence one needs to estimate the operator norm of  $K^0_N$ which can be conveniently bounded by the general bound on a matrix $K\in\CM^{N\times N}$ (following from the Schur test)
\begin{equation}
\label{eq-Schur}
\|K\|
\;\leq\;
\max\Big\{
\max_{n=1,\ldots,N} \sum_{l=1}^N |\langle n|K|l\rangle|
\,,\,
\max_{l=1,\ldots,N} \sum_{n=1}^N|\langle n|K|l\rangle|
\Big\}
\;.
\end{equation}
The matrix elements of $K^0_N$ are given by Lemma~\ref{lem-T0}, hence leading to
\begin{equation}
\label{eq-K0Bound}
\|K^0_N\|
\;\leq\;
\frac{N}{a}\max_{1\leq l\leq n\leq N'}
\Big\{\Big|\frac{\hat{\psi}_n}{\hat{\psi}_l}\Big|,\Big|\frac{\tilde{\psi}_n}{\tilde{\psi}_l}\Big|\Big\}
\;.
\end{equation}
This implies the claim.
\hfill $\Box$

\begin{rem}
{\rm
In the work \cite{LaSi}, the matrix elements of $K^\epsilon_N$ are simply bounded by
$$
|\langle n| K^\epsilon_N|l\rangle|
\;\leq\;
\frac{1}{a}\Big(|\epsilon|\frac{1}{a}\,+\,\frac{1}{a}\,+\, b\Big)^{n-l}
\;,
$$
where $t_n\in [a,b]$ with $a>0$. Then $\epsilon$ is simply bounded by an a priori bound on the spectrum of $H^N$ and this allows to show that there is some constants $\gamma>1$ and $C$ such that $|E^N_{n+1}-E^N_n|\geq C\gamma^{-N}$. For a random Jacobi matrix with Poisson statistics, this uniform bound is optimal. In the present setting of a (random) chiral Jacobi matrix, it will be shown that for typical configuration the bound can be considerably improved for the low lying eigenvalues. Let us also note that for a periodic system, the transfer matrices for energies inside the spectrum are conjugate to a rotation which, based on \eqref{eq-Schur}, allows to show $\|K^\epsilon_N\|\leq CN$, in turn leading to $|E^N_{n+1}-E^N_n|\geq \frac{1}{CN}$. 
}
\hfill $\diamond$
\end{rem}

\subsection{Bound on the gap above the smallest energy}
\label{sec-E1GapBounds}

This section proves a deterministic lower bound on the difference $E_2^N-E^N_1$ of the two smallest positive eigenvalues of a chiral Jacobi matrix $H^N$.

\begin{proposi}
\label{prop-LowerE2-1}
One has
\begin{equation}
\label{eq-LowerE2-1}
E_2^N-E^N_1
\;\geq\; 
\frac{a^2}{2bN^2} \Big(
\min_{1\leq l\leq n\leq N'}
\Big\{\Big|\frac{\hat{\psi}_l}{\hat{\psi}_n}\Big|,\Big|\frac{\tilde{\psi}_l}{\tilde{\psi}_n}\Big|\Big\}
\Big)^2
\;.
\end{equation}
\end{proposi}

\noindent {\bf Proof.} Let us apply Lemma~\ref{lem-LastSimon} for $E_n^N=E_1^N$ and $E_{n+1}^N=E_2^N$, with $\epsilon=E^N_1$. This gives
$$
E_2^N-E^N_1
\;\geq\; 
\frac{1}{\|K_N^\epsilon\|}
\;.
$$
The operator norm $\|K_N^\epsilon\|$ will again be bounded above by \eqref{eq-Schur}, with the matrix elements now given by \eqref{eq-KOpDef2} for $\epsilon=E^N_1$. For this purpose, the matrix elements of $K_N^\epsilon$ have to be computed. Let $\phi^{\epsilon,l}=(\phi^{\epsilon,l}_n)_{n\geq l}$ denote the solution of the eigenvalue equation with initial conditions $\phi^{\epsilon,l}_{l+1}=1$ and $\phi^{\epsilon,l}_{l}=0$. Then  in a similar manner as in \eqref{TM-Sol}
\begin{align}
\phi^{\epsilon,l}_n
\;=\;
\frac{1}{t_{l+1}}\;\binom{0}{1}^*
{\cal T}^\epsilon(n+1,l+1)
\binom{1}{0}
\;=\;
\langle n|K^\epsilon_N|l\rangle
\;,
\qquad
l\leq n
\;.
\label{TM-Sol2}
\end{align}
The states $\phi^{\epsilon,l}$ also solve the eigenvalue equation $H^{N,l}\phi^{\epsilon,l}=\epsilon \phi^{\epsilon,l}$ where $H^{N,l}$ is the restriction of $H^N$ to the sites $[l+1,\ldots,N]$, namely $H^{N,l}$ has Dirichlet boundary conditions at $l$ since $\phi^{\epsilon,l}_{l+1}=1$ and $\phi^{\epsilon,l}_{l}=0$. With these notations, $H^{N,0}=H^N$ and $\phi^{\epsilon,0}=\phi^\epsilon$. The smallest positive eigenvalue $E_1^{N,l}$ of $H^{N,l}$ satisfies $E_1^{N,l}\geq E_1^{N}$ by the interlacing theorem. Note that for odd size $N-l$ the Hamiltonian $H^{N,l}$ has a simple zero eigenvalue $E_0^{N,l}=0$. Again $\phi^{\epsilon,l}$ will be decomposed in even and odd sites, however, with the first site $l+1$ being considered as an odd site (such that Proposition~\ref{prop-EigEst} applies directly). More precisely, let us set
$$
\tilde{\phi}^{\epsilon,l}_{k}
\;=\;
\phi^{\epsilon,l}_{l+2k-1}
\;,
\qquad
\hat{\phi}^{\epsilon,l}_{k}
\;=\;
\phi^{\epsilon,l}_{l+2k}
\;,
$$
with integer $k\geq 1$ such that $l+2k-1\leq N'$ and $l+2k\leq N'$ respectively. Accordingly, the corresponding zero energy solutions will be defined by 
$$
\tilde{\psi}^l_k
\;=\;
\tilde{\phi}^{0,l}_{k}
\;.
$$
They can be expressed in terms of zero energy solution $\tilde{\psi}$ of the full Hamiltonian as defined in Section~\ref{sec-ZeroEnergy}, namely for even and odd $l$ respectively
$$
\tilde{\psi}^{2p}_k
\;=\;
\frac{\tilde{\psi}_{p+k}}{\tilde{\psi}_{p+1}}
\;,
\qquad
\tilde{\psi}^{2p+1}_k
\;=\;
\frac{\hat{\psi}_{p+k}}{\hat{\psi}_{p+1}}
\;.
$$

\vspace{.1cm}
\noindent Let us now apply the upper bound of \eqref{eq-LowestStateBound1} in Proposition~\ref{prop-EigEst} for the state $\phi^{\epsilon,l}$ (which is a solution of $H^{N,l}\phi^{\epsilon,l}=\epsilon \phi^{\epsilon,l}$ with Dirichlet boundary condition at $l+1$)  and $E=\epsilon\leq E_1^{N,l}$:
$$
|\phi^{\epsilon,2p}_{2p+2k-1}|
\;=\;
|\tilde{\phi}^{\epsilon,2p}_{k}|
\;\leq\;
|\tilde{\psi}^{2p}_{k}|
\;=\;
\Big|\frac{\tilde{\psi}_{p+k}}{\tilde{\psi}_{p+1}}\Big|
\;.
$$
Similarly,
$$
|\phi^{\epsilon,2p+1}_{2p+2k-1}|
\;=\;
|\tilde{\phi}^{\epsilon,2p+1}_{k}|
\;\leq\;
|\tilde{\psi}^{2p+1}_{k}|
\;=\;
\Big|\frac{\hat{\psi}_{p+k}}{\hat{\psi}_{p+1}}\Big|
\;.
$$
It only remains to bound $\phi^{\epsilon,2p+1}_{2p+2k}$ and $\phi^{\epsilon,2p}_{2p+2k}$. The three-term recurrence relations reads
$$
\epsilon \phi^{\epsilon,2p+1}_{2p+2k}
\;=\;
t_{2p+2k} \phi^{\epsilon,2p+1}_{2p+2k-1}+ t_{2p+2k+1}\phi^{\epsilon,2p+1}_{2p+2k+1}
\;.
$$
As in Section~\ref{sec-Osci}, this implies for the signs of the solutions
$$
\sgn(\phi^{\epsilon,2p+1}_{2p+2k})
\;=\;
\sgn(\phi^{\epsilon,2p+1}_{2p+2k-1})
\;=\;
-\,\sgn(\phi^{\epsilon,2p+1}_{2p+2k+1})
\;,
$$
and therefore
$$
\epsilon |\phi^{\epsilon,2p+1}_{2p+2k}|
\;=\;
t_{2p+2k} |\phi^{\epsilon,2p+1}_{2p+2k-1}|- t_{2p+2k+1}|\phi^{\epsilon,2p+1}_{2p+2k+1}|
\; \leq \; 
t_{2p+2k} |\phi^{\epsilon,2p+1}_{2p+2k-1}|
\; \leq \;  
t_{2p+2k} |\tilde{\psi}^{2p+1}_{k}|
\;.
$$
Combining this with \eqref{eq-E1K} for $\epsilon=E^N_1$ shows
$$
|\phi^{\epsilon,2p+1}_{2p+2k}|
\; \leq \; 
\frac{t_{2p+2k} |\tilde{\psi}^{2p+1}_{k}|}{\epsilon}
\; \leq \; 
2 t_{2p+2k} \|K^0_N\|\, |\tilde{\psi}^{2p+1}_{k}|
\;.
$$
Similarly, for $2p$ one gets:
$$
 |\phi^{\epsilon,2p}_{2p+2k}|\;  \leq \; 2 t_{2p+2k}\|K^0_N\|\, |\tilde{\psi}^{2p}_{k}|
\;.
$$
Collecting all these estimates in \eqref{eq-Schur} leads to the upper bound 
$$
\|K^\epsilon_N\|
\;\leq \;
\frac{N}{a} \max\{1,2b\|K^0_N\|\}\max_{1\leq l\leq n\leq N} 
\Big\{\Big|\frac{\hat{\psi}_n}{\hat{\psi}_l}\Big|,\Big|\frac{\tilde{\psi}_n}{\tilde{\psi}_l}\Big|\Big\}
\;,
$$
On the r.h.s., let now bound $\|K^0_N\|$ by \eqref{eq-K0Bound}. As $N\geq \frac{a}{2b}$ the first maximum is given by the second entry, and this directly implies the claim. 
\hfill $\Box$

\section{Probabilistic estimates for the random hopping model}
\label{sec-Probabilistic}

Section~\ref{sec-ChiralJac} provided deterministic spectral information on the chiral Hamiltonian $H^N$ for a given fixed configuration $(t_n)_{n=2,\ldots,N}$. In this section, the hoppings $(t_n)_{n=2,\ldots,N}$ are i.i.d. random variables drawn from an  interval $[a,b]$ with $a>0$ and $b<\infty$, with a distribution having a density.

\subsection{Random walk perspective}
\label{sec-RandomWalk}

Let us introduce the notations
\begin{equation}
\label{eq-RandWalkDef}
\tilde{\kappa}_n
\,=\,
\frac{t_{2n}}{t_{2n+1}}
\;,
\qquad
\tilde{w}_{n+1}\,=\,\tilde{w}_{n}\,+\,\log(\tilde{\kappa}_n)
\;,
\end{equation}
for $n=1,\ldots,N'-1$ and with initial condition $\tilde{w}_1=0$. Note that as $t_{2n}$ and $t_{2n+1}$ are i.i.d., the increments $\log(\tilde{\kappa}_n)$ are centered, actually the distribution is even symmetric. Moreover, by assumption the increments have an absolutely continuous distribution and supported by a finite interval. Let us denote their variance by
$$
\sigma^2\;=\;\EM(\log(\tilde{\kappa}_n)^2)
\;.
$$
Then \eqref{eq-RandWalkDef} defines a discrete-time, continuous step random walk $\tilde{w}=(\tilde{w}_n)_{n=1,\ldots,N'}$ with $w_n\in\RM$. In the following, we will consider $\tilde{w}$ as a given random configuration, rather than $(t_n)_{n=1,\ldots,N'}$. With these notations, the odd zero modes are given by
\begin{equation}
\label{eq-zeroE}
\tilde{\psi}_n\,=\,(-1)^{n-1}e^{\tilde{w}_{n}}
\;,
\qquad
n=1,\ldots,N'
\;.
\end{equation}
In the same manner, one can identify the random walk $\hat{w}$ associated with the even zero mode:
\begin{equation}
\label{eq-RandWalkDef2}
\hat{\kappa}_n
\,=\,
\frac{t_{2n+1}}{t_{2n+2}}
\;,
\qquad
\hat{w}_{n+1}\,=\,\hat{w}_{n}\,+\,\log(\hat{\kappa}_n)
\;,
\end{equation}
for $n=1,\ldots,N'-1$ and with initial condition $\hat{w}_1=0$, namely
\begin{equation}
\label{eq-zeroE2}
\hat{\psi}_n\,=\,(-1)^{n-1}e^{\hat{w}_{n}}
\;,
\qquad
n=1,\ldots,N'
\;.
\end{equation}
The identity \eqref{eq-2RandomWalks} implies 
\begin{equation}
\label{eq-HatTildeLink}
\hat{w}_n\;=\;-\tilde{w}_n\,+\,\log(t_2)-\log(t_{2n})
\;,
\end{equation}
so that the random walk $\hat{w}$ is essentially $-\tilde{w}$, up to an error of order $1$. To elaborate more on the bounds derived in Section~\ref{sec-ChiralJac}, let us introduce the following quantities:
$$
\tilde{W}_{\pm}\;=\;\pm\max\left\{\pm \tilde{w}_{n}\,:n=1,\ldots,{N'}\right\}
\;, 
\qquad 
\tilde{D}\,=\,
\tilde{w}_{N'}
\,=\,
\log
\left(
\prod^{N'-1}_{k=1}\frac{t_{2n}}{t_{2n+1}}
\right)
\;,
$$
notably $\tilde{W}_{\pm}$ are the maximal and the minimal value of the random walk $\tilde{w}$ and $\tilde{D}$ is the defect at the right edge. If $\tilde{D}=0$, one speaks of a (discrete) Brownian bridge configuration. By the central limit theorem, both $\tilde{W}_{\pm}$ and $\tilde{D}$ are typically of the order $\sqrt{N'}$. Indeed,
$$
\EM (\tilde{D}^2)
\;=\;
\mathbb{E}\Big(\big(\sum_{i=1}^{N'-1}\log(\tilde{\kappa}_i)\big)^2\Big)
\;=\;
(N'-1)\,\sigma^2
\;.
$$
Similarly, one associates $\hat{W}_{\pm}$ and $\hat{D}$ to $\hat{w}$. Due to \eqref{eq-HatTildeLink}, one has, with $c=\log(\frac{b}{a})>0$,
\begin{equation}
\label{eq-Wconnect}
|\tilde{W}_+ +\hat{W}_-|\,\leq \,c
\;,
\qquad
|\tilde{W}_- +\hat{W}_+|\,\leq \,c
\;,
\qquad
|\tilde{D}+\hat{D}|\,\leq\,c
\;.
\end{equation}
Furthermore, one also checks that $\EM (\tilde{D}^2)=\EM(\hat{D}^2)$. Using these notations for the extrema and the defects, one has the following deterministic statements about the low-lying eigenvalues $E^N_1$ and $E^N_2$:

\begin{proposi}
\label{prop-EnergyBoundOBC}
The smallest positive energy  $E^N_1$ satisfies the bounds
\begin{equation}
\label{eq-EnergyBound1OBC}
\frac{a^2}{2bN}\,e^{-(\tilde{W}_+-\tilde{W}_-)}
\;\leq\;
E^N_1
\;\leq\;
b\,
e^{-\max\{\tilde{W}_+-\tilde{D},|\tilde{W}_-|\}}
\;,
\end{equation}
and the gap above  $E^N_1$ is bounded below by
$$
E^N_2- E^N_1
\;\geq\;
\frac{a^4}{2b^3N^2}
\,
e^{-2(\tilde{W}_+-\tilde{W}_-)}
\;.
$$
\end{proposi}

\noindent {\bf Proof.} The lower bound in \eqref{eq-EnergyBound1OBC} follows from Proposition~\ref{prop-LowerE1} when it is combined with
\begin{equation}
\label{eq-MinQuotients}
\min_{1\leq l\leq n\leq N'}
\Big\{\Big|\frac{\hat{\psi}_l}{\hat{\psi}_n}\Big|,\Big|\frac{\tilde{\psi}_l}{\tilde{\psi}_n}\Big|\Big\}
\;\geq\;
e^{-\min\{\tilde{W}_+-\tilde{W}_-, \hat{W}_+-\hat{W}_-\}}
\;.
\end{equation}
Recall from \eqref{eq-Wconnect} that $\tilde{W}_++\hat{W}_-$ and $\hat{W}_++\tilde{W}_-$ are both bounded by a factor $\log(\frac{b}{a})$, so it follows that
$$
\min_{1\leq l\leq n\leq N'}
\Big\{\Big|\frac{\hat{\psi}_l}{\hat{\psi}_n}\Big|,\Big|\frac{\tilde{\psi}_l}{\tilde{\psi}_n}\Big|\Big\}
\;\geq\;
\frac{a}{b}\,e^{-(\tilde{W}_+-\tilde{W}_-)}
\;.
$$
For the proof of the upper bound, let us invoke Proposition~\ref{prop-MinMax}. The two bounds in \eqref{eq-MinMax} imply, when combined with the bounds $\langle\tilde{\psi}|\tilde{\psi}\rangle\geq e^{2\tilde{W}_+}$ and $\langle\hat{\psi}|\hat{\psi}\rangle\geq e^{2\hat{W}_+}$,
$$
E^N_1
\;\leq\; 
\frac{b\,e^{\tilde{D}}}{e^{\tilde{W}_+}}
\;,
\qquad
E^N_1
\;\leq\; 
\frac{b}{e^{-\hat{W}_+}}
\;.
$$
Taking the minimum of the two r.h.s. implies the claim. The lower bound on $E^N_2- E^N_1$ follows in a similar manner from Proposition~\ref{prop-LowerE2-1}, combined with \eqref{eq-Wconnect}.
\hfill $\Box$

\subsection{Probabilistic estimates on the smallest positive eigenvalue}
\label{sec-ProbEstE1}

In this section, it is illustrated how the random walk perspective allows to prove probabilistic spectral statements for the random hopping Hamiltonian. For this purpose, let us recall some well-known facts \cite{BS,Fel,Kal}. First of all, the central limit theorem states that the distribution of $N^{-\frac{1}{2}}\sigma^{-1}\tilde{D}$ is asymptotically Gaussian. Second of all, by Donsker's invariance principle the scaled random walk $N^{-\frac{1}{2}}\sigma^{-1}\tilde{w}_{[tN]}$ with continuous linear interpolations converges weakly in the uniform topology to a standard Brownian motion $(B_t)_{t\in[0,1]}$ (Wiener process). More precisely, for any continuous bounded real function $F$ on the space $(C([0,1]),\|\,.\,\|_\infty)$ one has convergence in expectation. The rescaled maximal value $N^{-\frac{1}{2}}\sigma^{-1}\tilde{W}_+$, in turn, converges to a half-normal distribution with density $\sqrt{\frac{2}{\pi}}\,e^{-\frac{x^2}{2}}\chi(x>0)$, as does $N^{-\frac{1}{2}}\sigma^{-1}|\tilde{W}_-|$. While this distribution has a finite density close to $x=0$, the positive random variable $\tilde{W}_+-\tilde{W}_-$ providing the lower bound in \eqref{eq-EnergyBound1OBC} follows the Feller distribution which has much less weight for small values. Of course, there are realization with much larger extremal values $\tilde{W}_+$ and $\tilde{W}_-$. However, this is untypical. More precisely, let us consider 
$$
\Ee_0
\;=\;
\Big\{\max\{\tilde{W}_+,|\tilde{W}_-|\} \leq \sigma {(N')}^{\frac{1}{2}+\delta}\Big\}
\;,
\qquad
\Ee'_0
\;=\;
\Big\{\min\{\tilde{W}_+,|\tilde{W}_-|\} \geq \sigma {(N')}^{\frac{1}{2}-\delta}\Big\}
\;.
$$
By Kolmogorov's maximal inequality
\begin{equation}
\label{eq-Kolmogorov}
\PM(\Ee_0)
\;\geq\;
1\,-\,\frac{\EM(\tilde{D}^2)}{\sigma^2 N^{1+2\delta}}
\;=\;
1\,-\,\frac{(N'-1)\EM(\log(\tilde{\kappa_i})^2)}{\sigma^2 (N')^{1+2\delta}}
\;\geq\;
1\,-\,(N')^{-2\delta}
\;.
\end{equation}
Moreover, fluctuation theory and renewal equations show that $\PM(\Ee'_0)\geq 1-C_1(N')^{-\delta}$ for some constant $C_1>1$ (a fact that will not be used later on). Hence $\PM(\Ee_0\cap\Ee'_0)\geq 1-2C_1(N')^{-\delta}$. On the set $\Ee_0\cap\Ee'_0$ of asymptotically full probability, Proposition~\ref{prop-EnergyBoundOBC} implies the following statement which hence shows what are the typical values of the smallest eigenvalues. 

\begin{coro}
\label{coro-EnergyBoundOBC}
There is a constant $C$ such that for all $\tilde{w}\in\Ee_0\cap\Ee'_0$ one has
\begin{equation}
\label{eq-EnergyBound1OBCcoro}
\frac{1}{CN}\,e^{-2\sigma N^{\frac{1}{2}+\delta}}
\;\leq\;
E^N_1
\;\leq\;
C\,e^{-\sigma N^{\frac{1}{2}-\delta}}
\;.
\end{equation}
Furthermore, for $\tilde{w}\in\Ee_0$
$$
E^N_2- E^N_1
\;\geq\;
\frac{a^4}{2b^3N^2}
\,
e^{-4\sigma N^{\frac{1}{2}+\delta}}
\;.
$$
\end{coro}

\begin{rem}
\label{rem-OBC1}
{\rm
There are, however, (untypical) configurations where both of the deterministic bounds in \eqref{eq-EnergyBound1OBC} are empty and the smallest energy does not satisfy $E^N_1\approx e^{-\sqrt{N}}$. For example, if the random walk is  increasing, one has $\tilde{D}=\tilde{W}^+$ and $\tilde{W}^-=0$, so that in this case, the bound \eqref{eq-EnergyBound1OBC} merely becomes $E^N_1\leq\,C$. Indeed, this happens for a periodic realization for which the eigenvector is a plane wave and the eigenvalue spacing is roughly $E^N_1\sim\frac{1}{N}$ so that there cannot be any subexponential decay. Furthermore, there are configurations for which one of the bounds in Proposition~\ref{prop-MinMax} is effective, and the other not. For example, if the random walk $\tilde{w}$ is positive (namely a meander), then typically $\tilde{W}_+-\tilde{D}\sim\sqrt{N}$ while $\hat{W}_+=0$. Another example is a Brownian bridge configuration, namely $\tilde{D}=0$, for which the respective size of $\tilde{W}_+$ and $|\tilde{W}_-|$ determines which bound in \eqref{eq-EnergyBound1OBC} is more effective. 
\hfill $\diamond$
}
\end{rem}

\subsection{Values of normalized zero modes at localization centers}
\label{sec-LocCenters}

The two zero modes $\tilde{\psi}$ and $\hat{\psi}$ are given by the exponential of the random walk, see \eqref{eq-zeroE} and \eqref{eq-zeroE2}. Hence these states are typically subexponentially localized with a localization centers $\argmax({\tilde{w}})$ and $\argmax({\hat{w}})\approx\argmin({\tilde{w}})$, respectively. While this is weaker than the usual exponential decay of Anderson localized states, one expects nevertheless that the normalized states $\tilde{\Psi}$ and $\hat{\Psi}$ defined in \eqref{eq-PsiNormalized} are of order $1$ near the localization center. This section shows that this holds true with large probability. As these states approximate the eigenfunction $\Phi^{E^N_1}$ of the smallest positive eigenfunction on parts of the sample by Proposition~\ref{prop-PhiPsiBound} and Corollary~\ref{coro-QuotientRandomWalkBound}, this also allows to deduce lower bounds on the $\Phi^{E^N_1}$. 

\vspace{.2cm}

On the mathematical side, the probabilistic lower bounds $\tilde{\Psi}$ and $\hat{\Psi}$ are direct corollaries of a remarkable (but nevertheless unnoticed) work by Davies and Kr\"amer \cite{DK}. In fact, the statement is somewhat surprising at first sight because, while the maximum  $\tilde{W}_+$ of $\tilde{w}$ is unique with probability $1$, values in an interval $[\tilde{W}_+-\delta,\tilde{W}_+)$ arise often, namely order of $\log(N)$ times (with prefactors depending on $\delta$). Nevertheless, the proof in \cite{DK} shows that in expectation these smaller values turn out to be summable.

\begin{proposi}
\label{prop-PsiLower}
Let $\tilde{n}=\argmax(\tilde{w})$. There is a constant $C_2$ such that, uniformly in $N$ and for any $\lambda>1$,
$$
\PM\Big(\Big\{ \tilde{\Psi}_{\tilde{n}}^2\geq\frac{1}{\lambda}\Big\}\Big)
\;\geq\;1\,-\,\frac{C_2}{\lambda}
\;.
$$
The same statement hold for $\hat{\Psi}_{\hat{n}}$ where $\hat{n}=\argmax(\hat{w})$.
\end{proposi}

\noindent {\bf Proof.} Let us look at probability of the complementary set:
$$
\PM\Big(\Big\{ \tilde{\Psi}_{\tilde{n}}^2<\frac{1}{\lambda}\Big\}\Big)
\;=\;
\PM\Big(\Big\{ 
\frac{\|\tilde{\psi}\|^2}{\tilde{\psi}_{\tilde{n}}^2}> \lambda\Big\}\Big)
\;\leq\;
\frac{1}{\lambda}\;
\EM
\Big(
\frac{\|\tilde{\psi}\|^2}{\tilde{\psi}_{\tilde{n}}^2}
\Big)
\;=\;
\frac{1}{\lambda}\;
\EM
\Big(
\sum_{n=1}^{N'}e^{2(\tilde{w}_n-\tilde{W}_+)}
\Big)
\;.
$$
Now equation (9) in \cite{DK} states that the expectation on the r.h.s. is bounded by a constant, uniformly in $N'$. This implies the claim.
\hfill $\Box$

\vspace{.2cm}

For the analysis in Section~\ref{sec-RealEnergies}, let us introduce the set
\begin{equation}
\label{eq-Ee1Def}
\Ee_1(\lambda)
\;=\;
\Big\{ \tilde{\Psi}_{\tilde{n}}^2\geq\frac{1}{\lambda}\Big\}
\cap
\Big\{ \hat{\Psi}_{\hat{n}}^2\geq\frac{1}{\lambda}\Big\}
\;.
\end{equation}
Then Proposition~\ref{prop-PsiLower} shows $\PM(\Ee_1(\lambda))\geq 1-\frac{2C_2}{\lambda}$.

\subsection{Locations of the extrema of the random walk}
\label{sec-RandomWalkExtrema}

Recall that $\text{argmax}(\tilde{w})$ and $\text{argmin}(\tilde{w})$ denote the location of the maximum $\tilde{W}_+$ and the minimum $\tilde{W}_-$ of the random walk $\tilde{w}=(\tilde{w}_n)_{n=1,\ldots, N'}$. As the distribution of the increments $\tilde{w}_n-\tilde{w}_{n-1}$ has a density, these locations are almost surely well-defined, namely almost surely the maximal and minimal values are only taken once. This section is about the event
%
$$
\Ee_2(0)
\;=\;
\left\{\text{argmax}(\tilde{w})\leq N''' \;\;\mbox{and}\;\;\text{argmin}(\tilde{w})=N''\right\}
\;,
$$
%
where $N''=\frac{N'}{2}=\frac{N}{4}$ and $N'''=\frac{N}{8}$. For sake of simplicity, $N$ will hence supposed to be divisible by $8$. The argument $0$ of $\Ee_2(0)$ is added for notational convenience because the event will be further generalized below. As it will be shown that the two extrema of $\tilde{w}$ can essentially be identified with two localization centers of the eigenfunction $\Phi^{E^N_1}$, these localization centers will be a distance $N'''$ of order $N$ apart for configurations in $\Ee_2(0)$. This favors disorder averaged quantum transport, provided that $\PM(\Ee_2(0))$ can be bounded from below. Such a lower bound is precisely what is proved below.

\vspace{.2cm}

Let us provide an intuitive understanding why $\PM(\Ee_2(0))\sim \frac{1}{N}$. First of all, the classical Sparre-Anderson theorem (see \cite[pp.~417]{Fel} or \cite{Kal}) states that, for the discrete time random walk $(\tilde{w}_n)_{n=1,\dots,N'}$ with bounded, centered increments with absolutely continuous distribution, one has $\PM(\text{argmin}(\tilde{w})=N'')\sim\frac{1}{N}$. Now $\mathcal{E}_2(0)$ requires information about both extrema. When they are far apart, one may expect them to be nearly independent so that the probability for both fixed extrema should be of order $\frac{1}{N^2}$. In $\Ee_2(0)$, however, one takes a union over $N'''$ such events, so that $\PM(\Ee_2(0))\approx N'''\frac{1}{N^2}\approx \frac{1}{N}$. The following result confirms this heuristics. We do not claim any novelty here, in fact there are numerous techniques that are available to obtain this and similar results \cite{MMS,SH}. The proof below uses the invariance theorem of Iglehard and Bolthausen \cite{Igl,Bol} and explicit computations for Brownian meanders \cite{BS}.

\begin{lemma}
\label{lem-E2Prob}
There is a constant $C_3>0$ such that
$$
\PM(\Ee_2(0))
\;\geq\; \frac{C_3}{N'}\; .
$$ 
\end{lemma}

\noindent \textbf{Proof.}
Define the random walks $l=(l_n)_{n=0,\ldots,N''}$ and $r=(r_n)_{n=0,\ldots,N''}$ to the left and right of $N''$ by
$$
l_n\,=\,\tilde{w}_{N''-n}-\tilde{w}_{N''}\;,
\qquad
r_n\,=\,\tilde{w}_{N''+n}-\tilde{w}_{N''}
\;.
$$
Then $\argmin(\tilde{w})=N''$ is equivalent to both $l$ and $r$ being meanders, namely random walks which are conditioned to be positive which will simply be denoted by $l\geq 0$ and $r\geq 0$. Thus
$$
\{\argmin(\tilde{w})=N''\}
\;=\;
\{l\geq 0\mbox{ and }r\geq 0\}
\;.
$$
Furthermore, let $L_+=\max(l)$ and $R_+=\max(r)$. Then
$$
\Ee_2(0)
\;=\;
\big\{l\geq 0\mbox{ and }r\geq 0\mbox{ and }L_+>R_+\mbox{ and }\argmax(l)\geq N''-N'''\big\}
\;.
$$
Another crucial insight is that the two random walks $l$ and $r$ are independent. Note that $N''-N'''=N'''$ which allows to reformulate the last condition in $\Ee_2(0)$. To render $\Ee_2(0)$ into a symmetric event, let us impose a second condition:
$$
\Ee_2(0)
\;\supset\;
\big\{l\geq 0\mbox{ and }r\geq 0\mbox{ and }L_+>R_+\mbox{ and }\argmax(l)\geq N'''\mbox{ and }\argmax(r)\geq N'''\big\}
\;.
$$
Now, by symmetry, both $L_+<R_+$ or $L_+>R_+$ have the same probability. Hence, using independence once the correlation $L_+>R_+$ is eliminated,
$$
\PM(\Ee_2(0))
\;\geq\;
\frac{1}{2}\;
\PM\big(
\{l\geq 0\mbox{ and }\argmax(l)\geq N'''\}\big)
\;
\PM\big(
\{r\geq 0\mbox{ and }\argmax(r)\geq N'''\}\big)
\;.
$$
Of course, the two factors on the r.h.s are equal, so let us focus on the first one. Next let us recall that by Iglehart-Bolthausen invariance principle the scaled meander $(N'')^{-\frac{1}{2}}\sigma^{-1}l_{[tN'']}$ with continuous linear interpolations converges weakly in the uniform topology to a standard Brownian meander $M=(M_t)_{t\in[0,1]}$ \cite{Igl,Bol}. For this process, the random variable $\theta=\argmax(M)$ follows the arcsin distribution $\PM(\theta>x)=\frac{2}{\pi}\arcsin(\sqrt{x})$ \cite{BS}.  In particular, $\PM(\theta>\frac{1}{2})=\frac{1}{2}$. Hence
\begin{align*}
\PM\big(\{l\geq 0\mbox{ and }\argmax(l)\geq N'''\}\big)
&
\;=\;
\PM\big(\{\argmax(l)\geq N''' \,\mid\,l\geq 0\}\big)
\PM\big(\{l\geq 0\}\big)
\\
&
\;=\;
\Big(\PM(\{\theta>\tfrac{1}{2}\})\,+\,o(N^{-1})\Big)
\,
2^{-N''}\binom{N''}{N'''}
\\
&
\;=\;
\Big(\frac{1}{2}\,+\,o(N^{-1})\Big)
\,
2^{-N''}\binom{N''}{N'''}
\\
&
\;=\;
\frac{1}{2}\,\frac{\sqrt{2}}{\sqrt{\pi}}\,\frac{1}{\sqrt{N''}}
\,+\,o(N^{-1})
\;,
\end{align*}
where the equality $\PM\big(\{l\geq 0\}\big)=2^{-N''}\binom{N''}{N'''}$ is the classical ballot theorem (which applies because the distribution of the increments is symmetric \cite{Fel}). Together one gets
$$
\PM(\Ee_2(0))
\;\geq\;
\frac{1}{4\pi}\,\frac{1}{N''}
\,+\,o(N^{-1})
\;,
$$
which implies the claim.
\hfill $\Box$

\vspace{.2cm}

For the proof of the main result of the paper, it is necessary to impose one further condition, namely that the absolute value of the minimum $\tilde{W}_-$ is larger than $\mu\sqrt{N''}$ for some constant $\mu\geq 0$:
\begin{equation}
\label{eq-Ee2bisDef}
\Ee_2(\mu)
\;=\;
\left\{\text{argmax}(\tilde{w})\leq N''' \mbox{ and }\text{argmin}(\tilde{w})=N''\mbox{ and }\tilde{W}_-< -\mu\sqrt{N''}\right\}
\;.
\end{equation}
For $\mu=0$, $\Ee_2(\mu)$ clearly reduces to $\Ee_2(0)$ as defined above. Nevertheless, as the extrema are typically of order $\sqrt{N''}$, one expects the probability of $\Ee_2(\mu)$ still to be of the order $\frac{1}{N'}$.

\begin{lemma}
\label{lem-E2bisProb}
There is a constant $C_4>0$ such that
$$
\PM(\Ee_2(\mu))
\;\geq\; \frac{C_4}{N'}\,e^{-\frac{\mu^2}{2\sigma^2}}\; .
$$ 
\end{lemma}

\noindent {\bf Proof.} The argument is an extension of the proof of Lemma~\ref{lem-E2Prob}. The main supplementary insight is that one has $l_{N''}=-\tilde{W}_-$ by construction, and that $l_{N''}$ is, under the condition $l\geq0$, the end point of a meander. For a standard Brownian meander, it is hence distributed according to the Rayleigh distribution $\rho(x)=x\,e^{-\frac{x^2}{2}}\chi(x>0)$ \cite{BS} so that the tail is given by $\int_{\mu}^\infty \rho(x)\,dx=e^{-\frac{\mu^2}{2}}$. After rescaling this leads to the supplementary factor $\exp(-\frac{\mu^2}{2\sigma^2})$. For the more detailed argument, let us first rewrite the event $\Ee_2(\mu)$ in terms of $l$ and $r$:
\begin{equation}
\label{eq-Ee2bisDef2}
\Ee_2(\mu)
\;=\;
\big\{l\geq 0\mbox{ and }r\geq 0\mbox{ and }L_+>R_+\mbox{ and }\argmax(l)\geq N'''\mbox{ and }l_{N''}\geq \mu\sqrt{N''} \big\}
\;.
\end{equation}
One can again symmetrize the event as above and then condition on being a meander:
\begin{align*}
\PM(\Ee_2(\mu))
&
\;\geq\;
\frac{1}{2}\;
\PM\big(
\{l\geq 0\mbox{ and }\argmax(l)\geq N''\mbox{ and } l_{N''}>\mu\sqrt{N'}\}\big)^2
\\
&
\;=\;
\frac{1}{2}\;
\PM\big(
\{\argmax(l)\geq N''\mbox{ and } l_{N''}>\mu\sqrt{N'}\}\,\big|\,l\geq 0\big)^2
\,\PM\big(\{l\geq 0\}\big)^2
\\
&
\;\geq\;
\frac{1}{2}\;
\Big(
\PM\big(
\{\argmax(M)\geq \tfrac{1}{2}\mbox{ and } M_1>\tfrac{\mu}{\sigma}\}\big)
\,+\,o(N^{-1})
\Big)^2\,
\frac{2}{\pi}\,\frac{1}{N''}
\;,
\end{align*}
where $M=(M_t)_{t\in[0,1]}$ is the standard Brownian meander and the Iglehart-Bolthausen invariance principle was again used in the last step, together with the bound $\PM\big(\{l\geq 0\}\big)^2\geq \frac{2}{\pi}\,\frac{1}{N''}$ which was already appealed to in the proof of Lemma~\ref{lem-E2Prob}. Now the event $\{\argmax(M)\geq \tfrac{1}{2}\}$ only depends on the shape of the meander $M$, while $\{M_1>\tfrac{\mu}{\sigma}\}$ depends on the scale. Hence these two events are independent (via the Vervaat transform) so that
$$
\PM\big(
\{\argmax(M)\geq \tfrac{1}{2}\mbox{ and } M_1>\tfrac{\mu}{\sigma}\}\big)
\;=\;
\PM\big(
\{\argmax(M)\geq \tfrac{1}{2}\}\big)
\,
\PM\big(\{ M_1>\tfrac{\mu}{\sigma}\}\big)
\;.
$$
As explained above, the last factor can be computed from the Rayleigh distribution. Recollecting all facts one readily concludes the proof.
\hfill $\Box$

\subsection{Lower bound on eigenfunction}
\label{sec-LowerEigenfun}

In this section it is shown that the eigenfunction $\hat{\Phi}^{E^N_1}$ satisfies a lower bound at its approximate localization center, at least for suitable configurations. For this purpose, let us again use the right random walk $r=(r_n)_{n=1,\ldots,N''}$ constructed in the proof of Lemma~\ref{lem-E2Prob} as well as its maximum value $R_+$. Then introduce the event
\begin{equation}
\label{eq-E3Def}
\Ee_3(\eta)
\;=\;
\big\{r\geq 0\mbox{ and }R_{+}\leq\eta\sqrt{N''}\big\}
\;,
\end{equation}
where $\eta>0$ is to be chosen later. Then using again the convergence to the Brownian meander $M=(M_t)_{t\in[0,1]}$ 
$$
\PM\big(\Ee_3(\eta)\big)
\;=\;
\big(
\PM\big(\big\{M\leq\eta \big\}\big)+o(N^{-1})\big)
\PM\big(\big\{r\geq 0\big\}\big)
\;\geq\;
\big(1-c'\, e^{-\frac{\eta^2}{2\sigma^2}}+o(N^{-1})\big)
\sqrt{\frac{2}{\pi N''}}
\;,
$$
for some finite constant $c'$. Now using the representation \eqref{eq-Ee2bisDef2} and $L_+\geq l_{N''}$, one finds that for $\mu>\eta$
$$
\Ee_2(\mu)\cap\Ee_3(\eta)
=
\big\{l\geq 0\mbox{ and }\argmax(l)\geq N'''\mbox{ and }l_{N''}\geq \mu\sqrt{N''} \big\}
\cap
\big\{r\geq 0\mbox{ and }R_+\leq\eta\sqrt{N''}\big\}
\,,
$$
because the condition $L_+>R_+$ is automatically satisfied. Hence the two events on the r.h.s. are independent and one can compute the probability of $\Ee_2(\mu)\cap\Ee_3(\eta)$ as the product of the probabilities of the two sets on the r.h.s.. Therefore, still for $\mu>\eta$, and using Lemma~\ref{lem-E2bisProb}
\begin{equation}
\label{eq-E2CapE3}
\PM\big(\Ee_2(\mu)\cap\Ee_3(\eta)\big)\,
\;\geq\;
\frac{C_4}{N'}\,e^{-\frac{\mu^2}{2\sigma^2}}\,\big(1-c'\,e^{-\frac{\eta^2}{2\sigma^2}}\big)
\;.
\end{equation}
%

\begin{proposi}
\label{prop-PhiLower}
For all $\tilde{w}\in\Ee_1(\lambda)\cap\Ee_2(\mu)\cap\Ee_3(\eta)$ with $\eta<\mu$ and $N$ sufficiently large, one has
$$
(\hat{\Phi}^{E^N_1}_{N''})^2 
\;\geq\;
\frac{1}{2\lambda}
\;.
$$
\end{proposi}

\noindent {\bf Proof.} Let us start out with Corollary~\ref{coro-QuotientRandomWalkBound} and use the definition of set $\Ee_1(\lambda)$:
\begin{align*}
(\hat{\Phi}^{E^N_1}_{N''})^2 
&
\;\geq\;
(\hat{\Psi}_{N''})^2
\Big(
1\,-\,\frac{1}{\|\hat{\psi}\|^2}
\sum_{p=1}^{N''}\tilde{\psi}_{N'-p+1}^2\sum_{k=1}^p\hat{\psi}_{N'-k+1}^2
\Big)^2
\\
&
\;\geq\;
\frac{1}{\lambda}
\Big(
1\,-\,\frac{1}{\|\hat{\psi}\|^2}
\sum_{p=1}^{N''}\tilde{\psi}_{N'-p+1}^2\sum_{k=1}^p\hat{\psi}_{N'-k+1}^2
\Big)^2
\;.
\end{align*}
Next note that
$$
\|\hat{\psi}\|^2\;\geq\; e^{2\hat{W}_+}\;\geq \;e^{-2c}e^{-2 \tilde{W}_-}
\;.
$$ 
Furthermore,
\begin{align*}
\sum_{p=1}^{N''}\sum_{k=1}^p\tilde{\psi}_{N'-p+1}^2 \hat{\psi}_{N'-k+1}^2
&
\;\leq\;
\,(N'')^2\,
\sup_{1\leq k\leq p\leq N''} 
\tilde{\psi}_{N'-p+1}^2 \hat{\psi}_{N'-k+1}^2
\\
&
\;\leq\;
\,(N'')^2\,
\sup_{1\leq k\leq p\leq N''} 
e^{2(\tilde{w}_{N'-p+1}+ \hat{w}_{N'-k+1})}
\\
&
\;\leq\;
\,(N'')^2\,e^{2c}
\sup_{1\leq k\leq p\leq N''} 
e^{2(\tilde{w}_{N'-p+1}- \tilde{w}_{N'-k+1})}
\\
&
\;\leq\;
\,(N'')^2\,e^{2c}\,e^{2 \tilde{\Delta}_{N''}}
\;,
\end{align*}
where $\tilde{\Delta}_l=\sup_{1\leq k\leq p\leq l}\tilde{w}_{N'-p+1}- \tilde{w}_{N'-k+1}$ is the biggest decrease of the random walk $\tilde{\psi}$ over the last $l$ sites. Then
$$
(\hat{\Phi}^{E^N_1}_{N''})^2 
\;\geq\;
\frac{1}{\lambda}
\Big(
1\,-\,
e^{2c}e^{2 \tilde{W}_-}
(N'')^2\,e^{2c}\,e^{2 \tilde{\Delta}_{N''}}
\Big)^2
\;\geq\;
\frac{1}{\lambda}
\Big(
1\,-\,
e^{4c}e^{-2\mu\sqrt{N''}}
(N'')^2\,e^{2 \tilde{\Delta}_{N''}}
\Big)^2
\;,
$$
where it was used that the configuration is taken from the set $\Ee_1(\mu)$. Now let us note that $\tilde{W}_-=w_{N''}$ only depends on the random variables in the first half, while $\tilde{\Delta}_{N''}$ only depends on the random walk on the right half. Deterministically one merely has the bound $\tilde{\Delta}_{N''}\leq \tilde{W}_+-\tilde{W}_-$ which is insufficient to show that the lower bound on $(\hat{\Phi}^{E^N_1}_{N''})^2$ stays positive. But for configurations in $\Ee_3(\eta)$, one has $\tilde{\Delta}_{N''}\leq R_+\leq \eta\sqrt{N''}$. Hence for configurations in $\Ee_2(\mu)\cap \Ee_3(\eta)$,
$$
(\hat{\Phi}^{E^N_1}_{N''})^2 
\;\geq\;
\frac{1}{\lambda}
\Big(
1\,-\,
e^{4c}e^{-2\mu\sqrt{N''}}
(N'')^2\,e^{2 \eta\sqrt{N''}}
\Big)^2
\;\geq\;
\frac{1}{2\lambda}
\;,
$$
for $N$ sufficiently large because by assumption $\mu>\eta$.
\hfill $\Box$

\section{Dynamics bound via lowest eigenstate}
\label{sec-RealEnergies}

In this section, the proof of Theorem~\ref{theo-Intro} is provided. Recall that the support of the random variables $t_n$ lies in $[a,b]$ with $a>0$ and $b<\infty$. Further recall that $M^N_q(T)$ is defined with the initial state $\phi(0)$ given by \eqref{eq-InitialState}. The chiral symmetry of $H^N$ allows to write $M^N_q(T)$ as the expectation value of a sum of two random variables involving the dynamics from initial state localized at just one site:
\begin{align*}
\tilde{m}_q^N(T)
&
\;= \;
\int_0^T \frac{dt}{T} \;
\langle N' | e^{\imath H^N t}
|X^N|^q e^{-\imath H^N t}|N'\rangle 
\;,
\\
\hat{m}_q^N(T)
&
\;= \;
\int_0^T \frac{dt}{T} \;
\langle N' +1| e^{\imath H^N t}
|X^N|^q e^{-\imath H^N t}|N'+1\rangle 
\;.
\end{align*}

\begin{lemma} 
\label{lem-BCH}
One has $M_q^N(T)=\frac{1}{2}\,\EM(\tilde{m}_q^N(T)+\hat{m}_q^N(T))$.
\end{lemma}

\noindent {\bf Proof:} Computing $M_q^N(T)$ gives, with $Y=|X^N|^q$,
$$
M_q(T)
\;=\;
\frac{1}{2}\,\EM\big(\tilde{m}_q^N(T)+\hat{m}_q^N(T)\big)
\,+\,
\EM\int_0^T\frac{dt}{T}\;\Re e\;\langle N'|Y(t)|N'+1\rangle
\;.
$$
To check that the last summand vanishes, let us use the (norm convergent) Baker-Campbell-Haussdorff formula in the form
$$
Y(t)
\;=\;
e^{\imath H^Nt}Ye^{-\imath H^Nt}
\;=\;
\sum_{n=0}^{\infty}\frac{(\imath t)^n}{n!}\,
[H^N,Y]_n
\;,
$$
where the iterated commutators are defined by
$$
[ H^N,Y]_n
\;=\;
[H^N,[ H^N,Y]_{n-1}]
\;,
\qquad
[ H^N,Y]_0
\;=\;
Y
\;.
$$
Replacing this sum in the above, one realizes that odd $n$ gives purely imaginary contributions which vanish when taking the real part. As to the summands with even $n$, they commute with $J^N$ because $H^N$ is odd and $Y=|X^N|^q$ is even w.r.t. $J$. But $J^N[ H^N,Y]_n J^N=[ H^N,Y]_n$ immediately implies that the matrix elements between even and odd sites vanish. 
\hfill $\Box$

\vspace{.2cm}

Let us next focus on proving a lower bound on $\tilde{m}_q^N(T)$. Suppose given a positive $C^{\infty}$-function $g\leq 1$ supported by $[0,1]$ satisfying $\int dx\,g(x) = \frac12$. Then for any $L\le N'$
\begin{align*}
\tilde{m}_q^N(T)
&
\;\geq\;
\sum_{|n-N'|> L} n^q
\int_0^T \frac{dt}{T} \;g\left(\frac{t}{T}\right)\;
\langle N' | e^{\imath H^N t}|n\rangle\langle n|e^{-\imath H^N t}|N'\rangle
\\
& \;> \;
L^q \int_0^T \frac{dt}{T} \;g\left(\frac{t}{T}\right)\,
 \sum_{|n-N'|> L} | \langle n| e^{-\imath H^N t} |N'\rangle|^2
\\
& 
\;= \;
L^q \int_0^T \frac{dt}{T} \;g\left(\frac{t}{T}\right)\,
 \left(1-\sum_{|n-N'|\le L} | \langle n| e^{-\imath H^N t} |N'\rangle|^2
\right)
\\
& 
\;= \;
L^q \int_0^T \frac{dt}{T} \;g\left(\frac{t}{T}\right)\,
 \left(1- \langle N'|e^{\imath H^N t}\Pi^Le^{-\imath H^N t} |N'\rangle
\right)
\;,
\end{align*}
where $\Pi^L=\sum_{|n-N'|\le L}|n\rangle\langle n|$ is the projection onto the central part $[N'-L,N'+L]\cap\ZM$ of the sample of size $N\geq 2L$. Next let $I$ be some energy interval centered at $0$. Later on, it will be chosen for each realization and thus depend on the disorder configuration. Let $I^c = \RM \setminus I$ denote the complementary set. Furthermore, let us set $\chi_{I}=\chi_{I}(H^N)$ and $\Pi^L(t)=e^{\imath H^N t}\Pi^Le^{-\imath H^N t}$. Splitting of the wave packet leads to
\begin{align*}
\tilde{m}_q^N(T) 
& 
\;> \; 
L^q \int_0^T \frac{dt}{T} \;g\left(\frac{t}{T}\right)
\,\Big( 
\langle N'| (\chi_{I}+ \chi_{I^c})|N' \rangle  
- 
\langle N'| (\chi_{I}+\chi_{I^c})\Pi^L(t)(\chi_{I}+\chi_{I^c})|N' \rangle
\Big)
\\
&
\;\geq\;
L^q \int_0^T \frac{dt}{T} \;g\left(\frac{t}{T}\right)
\,\Big( 
\langle N'| \chi_{I}|N' \rangle  
- 
\langle N'| \chi_{I}  \Pi^L(t)\chi_{I}|N' \rangle
-2\,\Re e \langle N'| \chi_{I} \Pi^L(t)\chi_{I^c}|N' \rangle
\Big)\;,
\end{align*}
because
$$
\langle N'| \chi_{I^c}|N' \rangle  
- 
\langle N'| \chi_{I^c}  \Pi^L(t)\chi_{I^c}|N' \rangle
\;=\;
\langle N'| \chi_{I^c} (\one- \Pi^L)(t)\chi_{I^c}|N' \rangle
\;\geq\;0
\;.
$$
The strategy in the following is to use time oscillations to bound $\langle N'| \chi_{I} \Pi^L(t)\chi_{I^c}|N' \rangle$ so that $\langle N'| \chi_{I}|N' \rangle  - \langle N'| \chi_{I}  \Pi^L(t)\chi_{I}|N' \rangle$ is the dominant positive term. For the latter purpose, choose $I$ sufficiently small so that it contains merely the smallest positive eigenvalue $E^N_1$, as then, with $\Phi^{E^N_1}$ denoting the corresponding normalized eigenstate,
\begin{align*}
\langle N'| \chi_{I}|N'\rangle 
- 
\langle N'| \chi_{I}  \Pi^L(t)\chi_{I}|N'\rangle
&
\;=\;
\langle N'| \chi_{I} \chi_{I}|N'\rangle 
- 
\langle N'| \chi_{I} e^{\imath E^N_1 t} \Pi^Le^{-\imath E^N_1 t}\chi_{I}|N'\rangle
\\
&
\;=\;
\langle N'| \chi_{I} (\one-\Pi^L)\chi_{I}|N'\rangle 
\\
&
\;=\;
|\langle N'|\Phi^{E^N_1}\rangle|^2
\langle \Phi^{E^N_1}|(\one-\Pi^L)|\Phi^{E^N_1}\rangle
\;.
\end{align*}
Further down it will be argued that the first factor $|\langle N'|\Phi^{E^N_1}\rangle|^2$ is large with a probability of order $\frac{1}{N}$ (due to Sparre-Anderson), and the second factor is large due to the high probability of separation between the extrema of the random walk for $L=\Oo(N)$. Of course, the two factors are not independent and a careful probabilistic analysis is needed. Note also that $I$ is still random (and has to be chosen such that there is just $E_1^N$ in it), and also the eigenfunctions are random at this point.

\vspace{.2cm}

For the treatment of the oscillatory contribution to $\tilde{m}_q^N(T)$,  let us next set
$$
A_{T,L}(I)
\;= \;
\int_0^T \frac{dt}{T}\; g\left(\frac{t}{T}\right) \,\langle N'|  \chi_I\Pi^L(t) \chi_{I^c}|N'\rangle
\;.
$$
Then the above becomes
\begin{align}
\tilde{m}_q^N(T)
& \;>\;
\frac{1}{2}\,L^q\, \Big( 
|\langle N'|\Phi^{E^N_1}\rangle|^2
\langle \Phi^{E^N_1}|(\one-\Pi^L)|\Phi^{E^N_1}\rangle
- 4\,|A_{T,L}(I)| \Big)
\;.
\label{eq-6.3}
\end{align}
This is a deterministic bound, and also the following estimate holds for every realization.

\begin{lemma} \label{Lemma61}
Suppose $I$ only contains one eigenvalue $E^N_1$ and let $p\in\NM$.  Denoting the $p$-th derivative of $g$ by $g^{(p)}$, one has
\begin{equation} 
\label{eq-6.4}
|A_{T,L}(I)|
\;\le \;
2\,L\,\max\big\{1,2\|g^{(p)}\|_{\infty}\big\}\,\frac{1}{1\,+\,\big(T (E_2^N-E^N_1)\big)^{p} }
\;.
\end{equation}
\end{lemma}

\noindent {\bf Proof:} 
Let us introduce the notation
$$
B_L(I)
\;=\;
\sum_{|n-N'|\le L} \langle N'| \chi_I(H^N) |N'\rangle
\,
\langle n| \chi_{I^c}(H^N)|n \rangle
\;.
$$
Define the complex spectral measures
$\mu_{l,k}$ by $\int\,d\mu_{l,k}(E) \,f(E) = \langle l|f(H^N)|
k\rangle$ for $f\in C_0(\RM)$, and write $\mu_k$ for $\mu_{k,k}$.
Then
\begin{equation}
\label{eq-6.5}
A_{T,L}(I)
\;=\;
\sum_{|n-N'|\le L} \int_I d\mu_{0,n}(E) \int_{I^c}
d\mu_{n,0}(E') \int_0^T \frac{dt}{T}\; g\left(\frac{t}{T}\right)
\;e^{\imath (E-E')t}
\;.
\end{equation}
Integrating by parts $p$ times gives
$$
\int_0^T
\frac{dt}{T}\; g\left(\frac{t}{T}\right) \;e^{\imath (E-E')t}
\;=
\;(-\imath T(E-E'))^{-p} \int_0^T \frac{dt}{T}
\;g^{(p)}\left(\frac{t}{T}\right) \;e^{\imath(E-E')t}
\mbox{ , }
$$
and thus, as $E=E^N_1$ and $E'\geq E^N_2$,
\begin{equation}
\label{eq-6.6}
|A_{T,L}(I)|
\;\leq\;
\min \Big\{
\frac{1}{2}\,,\, (T\cdot
(E^N_2-E^N_1))^{-p} \|g^{(p)}\|_{\infty} \Big\}
\sum_{|n-N'|\le L}
|\mu_{0,n}|(I) |\mu_{n,0}|(I^c)
\mbox{ , }
\end{equation}
where $|\mu_{j,k}|$ is the total variation of
$\mu_{j,k}$. For every Borel set $\Delta$, $|\mu_{j,k}|(\Delta) \le
\mu_j(\Delta)^{\frac{1}{2}} \mu_k(\Delta)^{\frac{1}{2}}$. Thus the Cauchy-Schwarz
inequality implies
\begin{align*}
\sum_{|n-N'|\le L} |\mu_{0,n}|(I) |\mu_{n,0}|(I^c)
\;\leq\; 
B_L(I)^{\frac{1}{2}} B_L(I^c)^{\frac{1}{2}}
\;.
\end{align*} 
Now using the rough bounds $\langle n| \chi_{I}(H^N)|n \rangle\leq 1$ and $\langle m| \chi_{I^c}(H^N)|m \rangle\leq 1$ for all $n,m$ to deduce
$$
|B_L(I)|
\;\leq \;
\sum_{|n-N'|\le L} 1
\;\leq \;
2L
\;,
\qquad
|B_L(I^c)|
\;\leq \;
2L
\;,
$$
one gets from (\ref{eq-6.6}).
$$
|A_{T,L}(I)|
\;\le \;
\min \left\{ \frac{1}{2}\,,\, \big(T
(E_2^N-E^N_1)\big)^{-p} \|g^{(p)}\|_{\infty} \right\}
\;2L
\;.
$$
Next one can use the bound
$$
\min\Big\{1,\frac{1}{x}\Big\}\,\leq\,\frac{2}{1+x}\;,
\qquad
x\geq 0
\;,
$$
to deduce 
$$
|A_{T,L}(I)|
\;\le \;
2\,L\, \frac{1}{1\,+\,
(E_2^N-E^N_1)^p\, T^p\, \frac{1}{2\|g^{(p)}\|_{\infty}} } 
\;,
$$
and finally, due to  $\frac{1}{1+\frac{x}{a}}\leq \max\{1,a\}\frac{1}{1+x}$ for $a>0$, the bound \eqref{eq-6.4} follows.
\hfill
\qed

\vspace{.2cm}

Replacing \eqref{eq-6.4} for $p=2$ into \eqref{eq-6.3} shows
\begin{align}
\tilde{m}_q^N(T)
& \;>\;
\frac{1}{2}\,L^q\, \left(
|\langle N'|\Phi^{E^N_1}\rangle|^2
\langle \Phi^{E^N_1}|(\one-\Pi^L)|\Phi^{E^N_1}\rangle
- 
C_5\,L\,\frac{1}{1\,+\,\big(T (E_2^N-E^N_1)\big)^{2} } \right)
\;,
\label{eq-MdetBound}
\end{align}
for $C_5=8\,\max\{1,2\|g^{(2)}\|_{\infty}\}$. The above above analysis directly carries over to $\hat{m}^N_q(T)$ providing the same lower bound \eqref{eq-MdetBound} with the state $|N'\rangle$ replaced by $|N'+1\rangle$. 

\vspace{.2cm}

Next let us condition on the set $\Ee_0$ which is overwhelming probability by \eqref{eq-Kolmogorov}. Due to Lemma~\ref{lem-BCH} and because the random variable is positive,
$$
M_q^N(T)
\;\geq\;
\frac{1}{2}\,\EM\big(\tilde{m}_q^N(T)+\hat{m}_q^N(T)\big|\Ee_0\big)
\;.
$$
Now the deterministic estimate \eqref{eq-MdetBound} and its counterpart for $\hat{m}^N_q(T)$ are replaced. Let us next focus on the second summand  in  \eqref{eq-MdetBound}, which actually also appears in $\hat{m}^N_q(T)$.

\begin{lemma} 
\label{lem-AverageA}
$$
\EM\Big(\frac{1}{1\,+\,T^2 (E_2^N-E^N_1)^{2} }\,\Big|\,\Ee_0\Big) 
\;\leq\;
\frac{1}{T^2}\;\frac{4b^6N^4}{a^8}\;e^{4\sigma N^{\frac{1}{2}+\delta}}
\;.
$$
\end{lemma}

\noindent {\bf Proof:} Due to Proposition~\ref{prop-EnergyBoundOBC} one has
$$
\EM\left(\frac{1}{(E^N_2-E^N_1)^2}\,\Big|\,\Ee_0\right)
\;\leq\;
\frac{4b^6N^4}{a^8}
\EM\big(e^{4(\tilde{W}_+-\tilde{W}_-)}\,\big|\,\Ee_0\big)
$$
Now the definition of the event $\Ee_0$ implies
$$
\EM\big(e^{4(\tilde{W}_+-\tilde{W}_-)}\,\big|\,\Ee_0\big)
\;\leq\; 
e^{4\sigma N^{\frac{1}{2}+\delta}}
\;.
$$
This directly gives the bound.
\hfill $\Box$

\vspace{.2cm}

Replacing this, one gets
\begin{align*}
M_q^N(T)
\;\geq\;&
\frac{1}{2}
\, 
\,\EM\Big(L^q
\big(|\langle N'|\Phi^{E^N_1}\rangle|^2+|\langle N'+1|\Phi^{E^N_1}\rangle|^2\big)
\langle \Phi^{E^N_1}|(\one-\Pi^L)|\Phi^{E^N_1}\rangle\,
\Big|\,\Ee_0\Big)
\\
&
\hspace{1cm}
-C_5\,L^{q+1}\,\frac{1}{T^2}\;\frac{4b^6N^4}{a^8}\;e^{4\sigma N^{\frac{1}{2}+\delta}}
\;.
\end{align*}
The time $T$ will later on be chosen such that the error term is small. Furthermore, let us choose $L=\frac{N}{4}$ from now on and recall $N''=\frac{N'}{2}$ as well as $N'''=\frac{N''}{2}$. In the last expression, let us now first focus on the following of the two positive summands:
$$
Q_N\;=\;
\EM\Big(
|\langle N'+1|\Phi^{E^N_1}\rangle|^2 \langle \Phi^{E^N_1}|(\one-\Pi^L)|\Phi^{E^N_1}\rangle\,
\Big|\,\Ee_0\Big)
\;.
$$
Then
\begin{align*}
Q_N
&
\;=\;
\EM\Big(
(\hat{\Phi}^{E^N_1}_{N''})^2 
\sum_{l=1}^{N'''}\,(\tilde{\Phi}^{E^N_1}_l)^2
\,\Big|\,\Ee_0\Big)
\;\geq\;
\sum_{l=1}^{N'''}\,
\EM\Big(
(\hat{\Phi}^{E^N_1}_{N''})^2 
(\tilde{\Psi}_l)^2\,
\Big|\,\Ee_0\Big)
\;\geq\;
\EM\Big(
(\hat{\Phi}^{E^N_1}_{N''})^2 
(\tilde{\Psi}_{\tilde{n}})^2
\,\Big|\,\Ee_0\Big)
\;,
\end{align*}
where the first inequality follows from Proposition~\ref{prop-PhiPsiBound}. Next let us consider the events $\Ee_1(\lambda)$, $\Ee_2(\mu)$ and $\Ee_3(\eta)$ introduced in \eqref{eq-Ee1Def}, \eqref{eq-Ee2bisDef} and \eqref{eq-E3Def} respectively. Also conditioning on these events for $\eta<\mu$ and setting 
$$
\Ee(\lambda,\mu,\eta)
\;=\;
\Ee_0\cap\Ee_1(\lambda)\cap\Ee_2(\mu)\cap\Ee_3(\eta)
\;,
$$
one has after appealing to Proposition~\ref{prop-PhiLower}
$$
Q_N
\;\geq\;
\EM\Big(
(\hat{\Phi}^{E^N_1}_{N''})^2 
(\tilde{\Psi}_{\tilde{n}})^2
\,\Big|\,\Ee(\lambda,\mu,\eta)\Big)
\;\geq\;
\frac{1}{\lambda}\;
\EM\Big(
(\hat{\Phi}^{E^N_1}_{N''})^2 
\,\Big|\,\Ee(\lambda,\mu,\eta)\Big)
\;\geq\;
\frac{1}{2\lambda^2}\;
\PM\big(\Ee(\lambda,\mu,\eta)\big)
\;.
$$
Now both $\Ee_0$ and $\Ee_1(\lambda)$ are of overwhelming probability for $\lambda$ sufficiently large, and the probability of $\Ee_2(\mu)\cap\Ee_3(\eta)$ is bounded from below in \eqref{eq-E2CapE3}.

\vspace{0.2cm}

\begin{figure}
\includegraphics[width=0.49\textwidth]{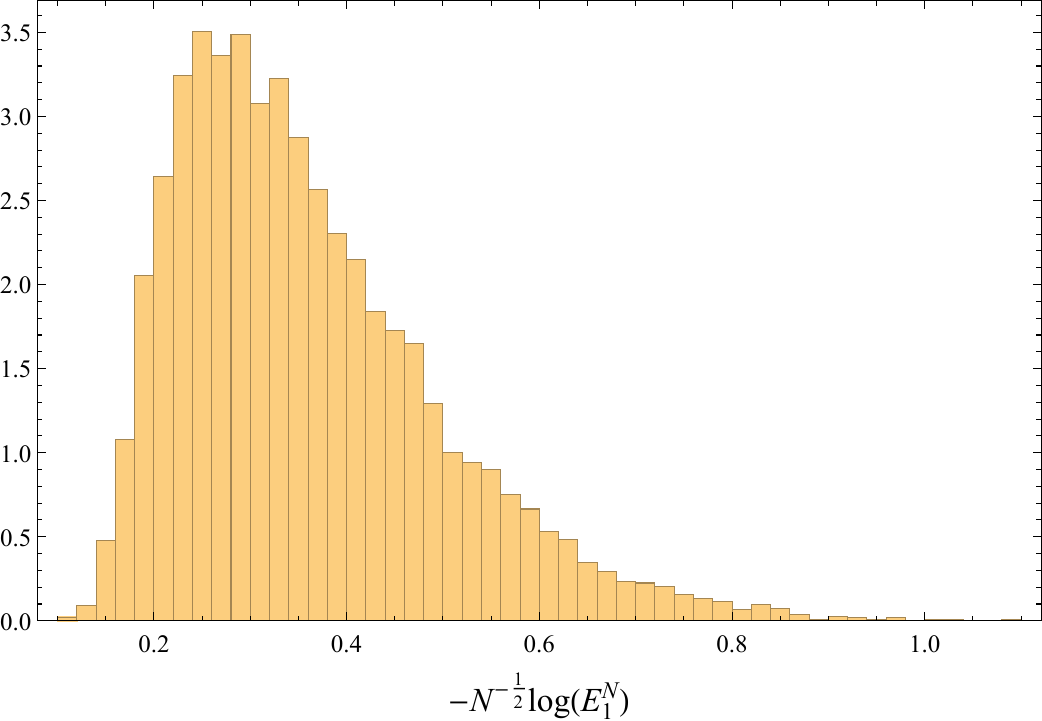}
\includegraphics[width=0.49\textwidth]{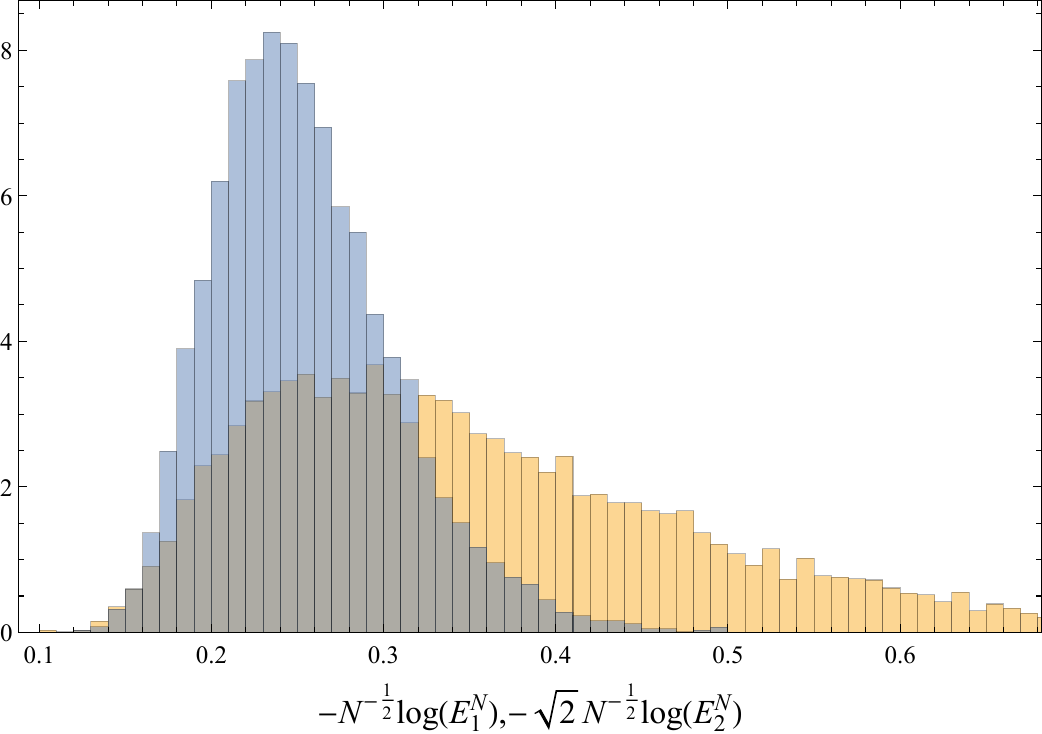}
\caption{\textit{The first histogram shows the suitably scaled logarithm of the lowest energy $E^N_1$ for $N=1000$ and $10000$ samples. The distribution of the hopping parameters are drawn uniformly from $[1.7,4.5]$. The second plot shows also shows the distribution of the second eigenvalue $E^N_2$ {\rm (}in blue{\rm )}, taking into account the factor $\sqrt{2}$ in \eqref{eq-EnergyBehav}.
\label{fig:wrapfig}
}}
\end{figure}

\noindent \textbf{Proof} of Theorem~\ref{theo-Intro}. By Lemma~\ref{lem-BCH}
$$
M_q^N(T)
\;\geq\;
\frac{1}{2}\,\EM\big(
\max\{\tilde{m}_q^N(T),\hat{m}_q^N(T)\}
\big)
\;.
$$
Let $\alpha\geq\frac{1}{2}+\delta$ and $a^4T\geq 2b^3 e^{(2+\delta')\sigma N^{\alpha}} $, for some $\delta'>0$ that does not scale with $N$. Combining the above results, one finds
\begin{align*}
M^N_q(T)
&
\;\geq\;
\frac{1}{4\lambda^2}\,\PM\big(\Ee(\lambda,\mu,\eta)\big)\,(N''')^q\;-\;
\frac{C_5}{2}\,(N''')^{q+1}\,N^4\,e^{-2\delta' \sigma N^{\alpha}}
\\
&
\;=\; 
A_q(\lambda,\mu,\eta)N^{q-1}
\,-\,\mathcal{O}\big(N^{q+5}e^{-2\delta' \sigma N^{\alpha}}\big)
\;,
\end{align*}
where $A_q(\lambda,\mu,\eta)=N\frac{8^{-q}}{4\lambda^2}\PM\big(\Ee(\lambda,\mu,\eta)\big)$ does not grow  with $N$ due to \eqref{eq-E2CapE3}. Bounding $T$ from above by $a^4T\leq 2b^3e^{(2+\delta'')\sigma N^{\alpha}}$ for some $\delta''>\delta'$, this leads to
$$
M^N_q(T)
\;\geq\; 
A_{q,\alpha}(\lambda,\mu,\eta) \log(aT)^{\frac{q-1}{\alpha}}
\,-\,
\mathcal{O}\big(N^{q+5}e^{-2\delta' \sigma N^{\alpha}}\big)
\;,
$$
with $A_{q,\alpha}(\lambda,\mu,\eta)=A_q(\lambda,\mu,\eta)\left(\frac{1}{(2+\delta'')\sigma}\log\left(\frac{a^3}{2b^3}\right)\right)^\frac{q-1}{\alpha}$.
\hfill $\Box$

\begin{figure}
\includegraphics[width=0.49\textwidth]{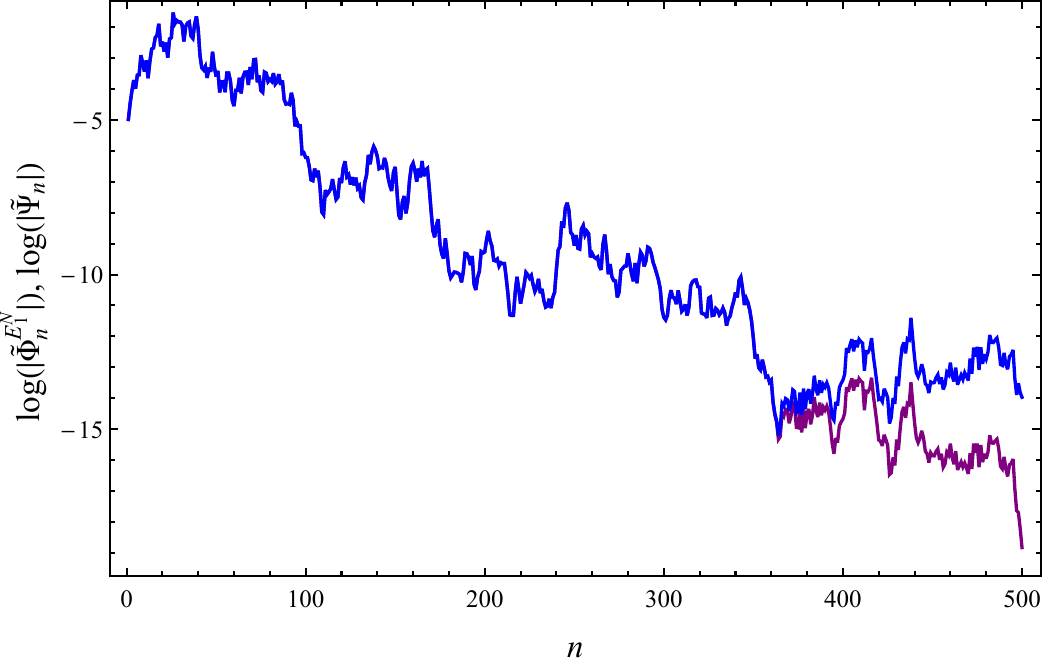}
\includegraphics[width=0.49\textwidth]{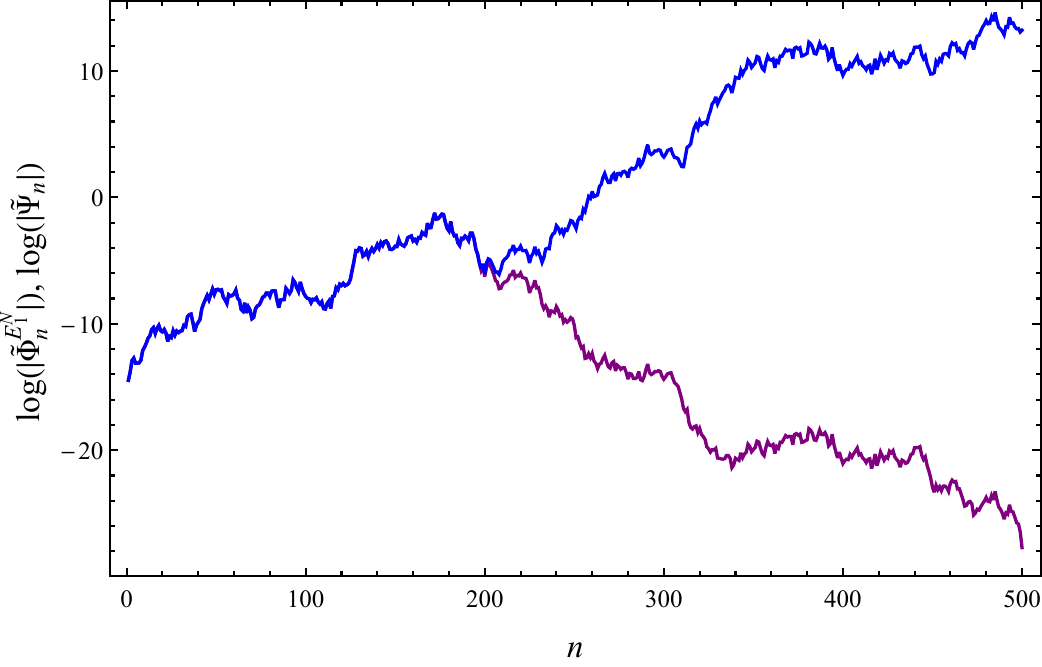}
\caption{\textit{Comparison of $\log(|\tilde{\Phi}^{E^N_1}_{n}|)$ {\rm (}in purple{\rm )} with $\log(|\tilde{\Psi}_n|)= \tilde{w}_n-\log(\sqrt{2}\|\tilde{\psi}\|)$ {\rm (}in blue{\rm )} for two different realizations again with $N=1000$. In the one in the left  plot, the position of the maximum of the random walk is smaller than the position of the minimum of the random walk. The sample is actually the same one as in {\rm Figure~\ref{fig-intro}}. In the second one, the position of the maximum is larger than that of the minimum.  The random walk {\rm (}still in blue{\rm )} is shifted by $c$ to $\tilde{w}_n+c$ with $c$ chosen such that the values at $n=1$ coincide.
}}
\label{fig:wrapfig2}
\end{figure}

\section{Numerical illustrations and comments}
\label{sec-NumIll}

This brief section provides some supplementary numerical information about the small eigenvalues and the eigenfunction of the smallest eigenvalue.  Let us start out by considering the positive eigenvalues $0<E^N_1<E^N_2<\ldots < E^N_k<\ldots$ of $H^N$. If one supposes that the small positive eigenvalues of $H^N$ are roughly distributed according to $\Nn$, one deduces from the Dyson peak behavior in \eqref{eq-Asymptotics} that
\begin{equation}
\label{eq-EnergyBehav}
\frac{k}{2N}\;\approx\;\frac{C}{\ln(E^N_k)^2}
\quad
\Longleftrightarrow
\quad
E^N_k\;\approx\;e^{-\frac{1}{\sqrt{k}}\sqrt{2CN}}
\;.
\end{equation}
For the smallest eigenvalue $E^N_1$, namely $k=1$, this roughly agrees with Corollary~\ref{coro-EnergyBoundOBC}. The factor $1/\sqrt{k}$ indicates what one may expect for the next eigenvalues.  Of course, all eigenvalues are random. Figure~\ref{fig:wrapfig} shows the distribution of $-N^{-\frac{1}{2}} \log(E_1^N)$ and, for sake of comparison, that of $-\sqrt{2} N^{-\frac{1}{2}} \log(E_2^N)$. It roughly confirms \eqref{eq-EnergyBehav}, but shows that these quantities have non-trivial distributions, implying that the smallest energies themselves have giant fluctuations. The numerical method is based on oscillation theory for chiral Jacobi matrices which will be discussed in detail in a future work. Let us point out though that this method is considerably more efficient than sparse exact diagonalization techniques. It is an interesting problem to determine the distribution of $-N^{-\frac{1}{2}} \log(E_1^N)$. Numerics indicate that is depends on the variance of $\log(\tilde{\kappa})$, but not the details of the distribution.

\vspace{.2cm}

Next let us address the normalized eigenfunction $\Phi^{E^N_1}$ and focus on the component $\tilde{\Phi}^{E^N_1}$. It is one of the key insights of Section~\ref{sec-Probabilistic}, based on the deterministic Corollary~\ref{coro-QuotientRandomWalkBound}, that $\tilde{\Phi}^{E^N_1}$ is well-approximated by the normalized zero mode $\tilde{\Psi}$ whenever the position of the maximum of the random walk $\tilde{w}$  is smaller than the position of the minimum. This happens with probability $\frac{1}{2}$. The left plot of Figure~\ref{fig:wrapfig2} shows such a favorable configuration and indeed one sees excellent agreement with $\log(|\tilde{\Phi}^{E^N_1}|)$, except at the right end of the sample. Note, however, that the state has negligible weight on this part of the sample. Nevertheless, this deviation is crucial because $\tilde{\Phi}^{E^N_1}$ has to satisfy the right boundary condition, which $\tilde{\Psi}$ does not do. As will be explained elsewhere, this becomes particularly transparent in oscillation theory.  In the second sample considered in the right plot of Figure~\ref{fig:wrapfig2}, the minimum lies to the left of the maximum. Based on Corollary~\ref{coro-QuotientRandomWalkBound}, one does not expect  good agreement of $\tilde{\Phi}^{E^N_1}$ with $\tilde{\Psi}$, and this is confirmed by the plot. Nevertheless, after a shift of the random walk by a suitable constant, one does have good agreement of $\tilde{\Phi}^{E^N_1}$ with the shifted random walk on the left side of the sample. Hence the zero mode reproduces the smallest eigenvector very well locally. Of course, this comes as no surprise because both sequences satisfy almost the same recurrence relation, so locally near the left boundary they only differ by a factor which depends on the normalization constants. Summing up, for roughly half of the realizations (those with the maximum appearing before the minimum) there is a good agreement of the eigenfunction with the zero energy mode on the relevant parts of the sample where there is most of the weight of the normalized states.

\vspace{.2cm}

Based on this latter fact, one can use probabilistic information on random walks for the analysis of the eigenfunction of the smallest eigenvalue. One of the key facts, proved in Proposition~\ref{prop-PhiLower} and used in Section~\ref{sec-RealEnergies}, is that the maximal value of the normalized eigenfunction $\Phi^{E^N_1}$ (taken at the localization center) is of order $1$ with high probability. Note that the probability of the set in Proposition~\ref{prop-PhiLower} is only of order $\frac{1}{N}$, but this factor is due to the fact that the localization center is fixed precisely at the center, which by the Sparre-Anderson theorem gives the factor $\frac{1}{N}$. Taking a sum over all possible localization centers indeed shows that the maximum of $\Phi^{E^N_1}$  is larger than a positive constants with positive probability. Figure~\ref{fig:wrapfig3} shows the distribution of this maximal value, and the second plot compares it to the distribution of the maximal value of the normalized zero energy state $\tilde{\Psi}$. One sees that there are merely small discrepancies, even though no conditioning on the good configurations (with maximum before minimum) was taken. 

\vspace{.2cm}

Finally, let us compare the bound of Theorem~\ref{theo-Intro} with the predictions of \cite{BAK} which states that $x\sim \log(t)^2$, see eq.~(1) in \cite{BAK}. In terms of the second instead of first moments, this corresponds to  $M_2^N(T)\sim \log(T)^4$, with the time scales being $T\sim e^{\sqrt{N}}$, both here and in \cite{BAK}. The lower bound in Theorem~\ref{theo-Intro} merely gives $M_2^N(T)\geq C \log(T)^2$ and is hence weaker. In fact, the proof in Section~\ref{sec-RealEnergies} only uses the one diagonal contribution \eqref{eq-heuristics} stemming from the smallest energy $E^N_1$. Improving the bound would require showing that the $N^\beta$ states with the smallest energies contribute to the transport. If this contribution is the same as by the lowest energy state (which is definitely and overestimate), one could then conclude that $M_q^N(T)\geq C N^{\beta} \log(T)^{2(q-1)}=C\log(T)^{2(q-1+\beta)}$, which only for $\beta=1$ and hence a finite fraction of contributing states would confirm the prediction of \cite{BAK}. In our opinion, this is not realistic and we expect that much fewer states can be expected to contribute to quantum transport.

\begin{figure}
\includegraphics[width=0.49\textwidth]{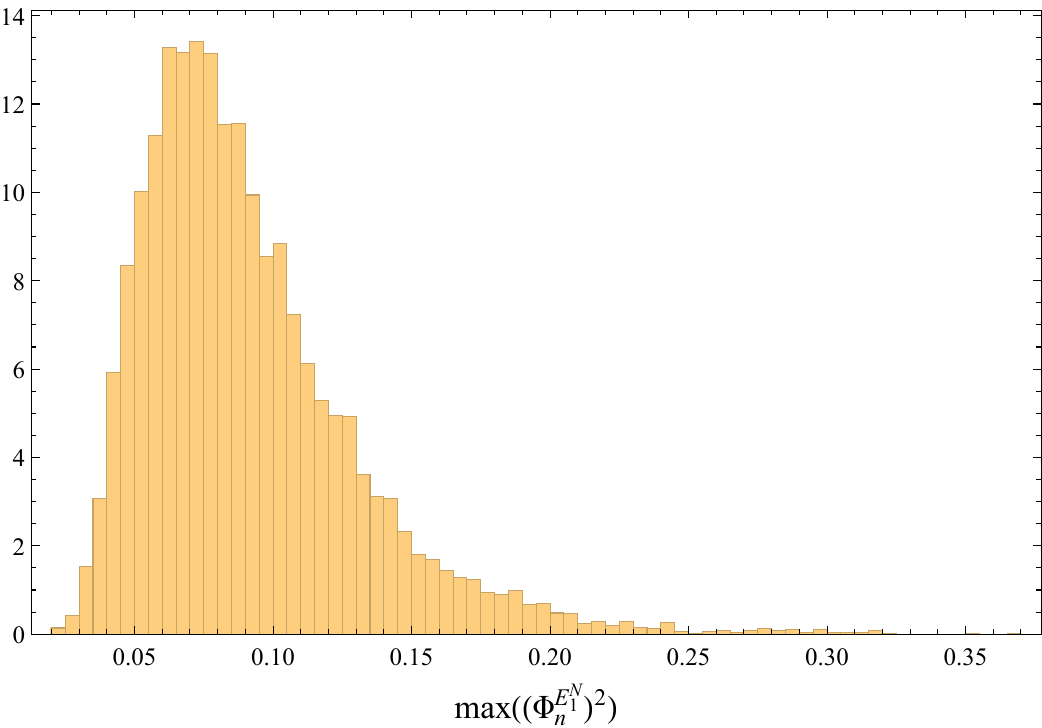}
\includegraphics[width=0.49\textwidth]{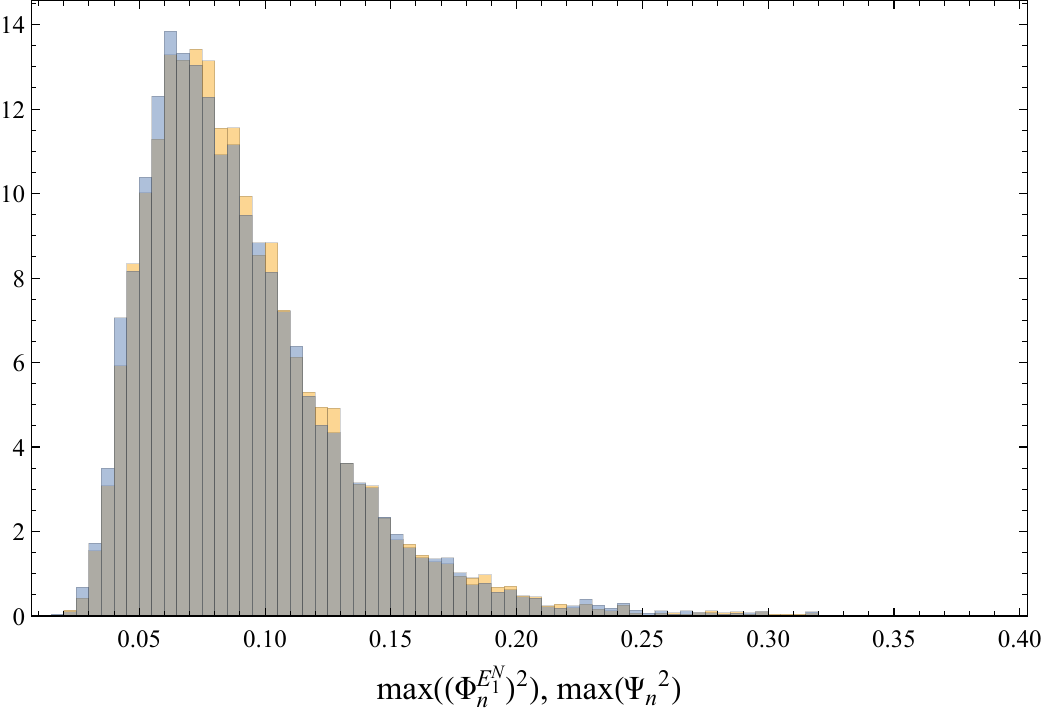}
\caption{\textit{Histogram of the value of the normalized eigenstate $\Phi^{E^N_1}$ at the localization center for $10000$ samples (computed with the Arnoldi algorithm) and, in the second plot, comparison with the distribution with the maximal value of the normalized random walk $\tilde{\Psi}$.
}}
\label{fig:wrapfig3}
\end{figure}



\begin{thebibliography}{99}

\bibitem{BAK} D.~Bagrets, A.~Altland, A.~Kamenev, {\sl Sinai Diffusion at Quasi-1D Topological Phase Transitions}, Phys. Rev. Lett. {\bf 117}, 196801 (2016).

\bibitem{BF} L.~Balents,  M.~P.~A.~Fisher, {\sl Delocalization transition via supersymmetry in one dimension}, Phys. Rev. {\bf B 56}, 12970 (1997).

\bibitem{BSS} S.~Barkhofen, D.~Syamsundar, S.~Sperling, C.~Silberhorn, A.~Altland, D.~Bagrets, K.~W.~Kim, T.~Micklitz, {\sl Experimental observation of topological quantum criticality}, Phys. Rev. Research {\bf 6},  033194 (2024).

\bibitem{BG} S.~de~Bi{\`e}vre, F.~Germinet, {\sl Dynamical Localization for the Random Dimer Schr{\"o}dinger Operator}, J. Stat. Phys. {\bf 98}, 1135-1148 (2000).

\bibitem{Bol} E.~Bolthausen, {\sl On a functional central limit theorem for random walks conditioned to stay
positive}, Ann. Prob. {\bf 4},  480-485  (1976).

\bibitem{BS} A.~N.~Borodin, P.~Salminen, {\sl Handbook of Brownian motion-facts and formulae}, 2nd corrected printing of 2nd edition, (Springer, Basel, 2015).

\bibitem{DK} P.~L.~Davies, W.~Kr\"amer, {\sl The Dickey-Fuller test for exponential random walks}, Econometric Theory {\bf 19}, 865-877 (2003).

\bibitem{DSS} J.~De Moor, C.~Sadel, H.~Schulz-Baldes, {\sl Footprint of a topological phase transition on the density of states}, Lett. Math. Phys. {\bf 113}, 96 (2023).

\bibitem{DSS2} J.~De Moor, C.~Sadel, H.~Schulz-Baldes, {\sl Scaling of the Lyapunov exponent at a balanced hyperbolic critical point}, Annales H. Poincar\'e {\bf 27}, 2199-2242 (2026).

\bibitem{DKS} M.~Drabkin, W.~Kirsch, H.~Schulz-Baldes, {\sl Transport in the random Kronig-Penney model}, J. Math. Phys. {\bf 53}, 122109 (2012).

\bibitem{DWP} D.~H.~Dunlap, H.-L.~Wu, P.~W.~Phillips, {\sl Absence of Localization in  Random-Dimer Model}, Phys. Rev. Lett. {\bf 65}, 88-91 (1990).

\bibitem{Dys} F.~J.~Dyson,  {\sl The dynamics of a disordered linear chain}, Phys. Rev. {\bf 92}  1331-1334 (1953).

\bibitem{EM} F.~Evers, A.~D.~Mirlin, {\sl Anderson transitions}, Rev. Mod. Phys. {\bf 80}, 1355-1417 (2008).
 
\bibitem{Fel} W.~Feller, {\sl  An introduction to probability theory and its applications, Volume 2}, (John Wiley \& Sons, Hoboken NJ, 1991).
 
\bibitem{Fis} D.~S.~Fisher, {\sl Critical behavior of random transverse-field Ising spin chains}, Phys. Rev. {\bf B 51}, 6411 (1995).

\bibitem{GK} F.~Germinet, F.~Klopp, {\sl Spectral statistics for random Schr\"odinger operators in the localized regime}, J. European Math. Soc. {\bf 16}, 1967-2031 (2014).

\bibitem{Gua} I.~Guarneri, {\sl Singular continuous spectra and discrete wave packet dynamics}, J. Math. Phys. {\bf 37}, 5195-5206  (1996).

\bibitem{HJ} R.~Hayn,  W.~John, {\sl Effective equations for disordered one-dimensional systems}, Zeitschrift f\"ur Physik B Cond. Mat. {\bf 67}, 169-177  (1987).

\bibitem{Igl} D.~L.~Iglehart, {\sl Functional central limit theorems for random walks conditioned to stay positive}, Annals Prob. {\bf 2}, 608-619 (1974).

\bibitem{IM} F.~Igl\'oi, C.~Monthus, {\sl Strong disorder RG approach of random systems}, Physics Reports {\bf 412}, 277-431 (2005).

\bibitem{JLM} S.~Jitomirskaya, W.~Liu, L.~Mi, {\sl Sharp palindromic criterion for semi-uniform dynamical localization}, {\tt arXiv:2410.21700}.

\bibitem{JSS} S.~Jitomirskaya, H.~Schulz-Baldes, G.~Stolz, {\sl Delocalization in random polymer models}, Commun. Math. Phys. {\bf 233}, 27-48  (2003).

\bibitem{JS} S.~Jitomirskaya, H.~Schulz-Baldes,  {\sl Upper bounds on wavepacket spreading for random Jacobi matrices}, Commun. Math. Phys. {\bf 273}, 601-618  (2007).

\bibitem{Kal} O.~Kallenberg, {\sl Foundations of modern probability}, 2nd Edition, (Springer, New York, 2002). 

\bibitem{KHQ} I.~Komissarov, T.~Holder, R.~Queiroz, {\sl Quantum critical dynamics induced by topological zero modes}, Phys. Rev. Lett. {\bf 136},  136602 (2026).

\bibitem{KV} M.~Kotowski, B.~Vir\'ag, {\sl Dyson's spike for random Schr\"odinger operators and Novikov-Shubin invariants of groups}, Commun. Math. Phys. {\bf 352},  905-933 (2017).

\bibitem{LaSi} Y.~Last, B.~Simon, {\sl Fine structure of the zeros of orthogonal polynomials, IV. A priori bounds and clock behavior}, Comm. Pure Appl. Math. {\bf 61}, 486-538  (2008).

\bibitem{MHMD} J.~Mard, J.~A.~Hoyos, E.~Miranda, V.~Dobrosavljevi\'c, {\sl Strong-disorder renormalization-group study of the one-dimensional tight-binding model}, Phys. Rev.  {\bf B 90},  125141 (2014).

\bibitem{McK} R.~H.~McKenzie, {\sl Exact results for quantum phase transitions in random XY spin chains}, Phys. Rev. Lett. {\bf 77}, 4804 (1996).

\bibitem{MSHP} I.~Mondragon-Shem, J.~Song, T.~L. Hughes,  E.~Prodan, {\sl  Topological criticality in the chiral-symmetric AIII class at strong  disorder}, Phys. Rev. Lett. {\bf 113},  046802 (2014).

\bibitem{MMS} F.~Mori, S.~N.~Majumdar, G.~Schehr, {\sl Distribution of the time between maximum and minimum of random walks}, Phys. Rev. {\bf  E 101}, 052111 (2020).

\bibitem{Nak} F.~Nakano, {\sl Distribution of localization centers in some discrete random systems}, Rev. Math. Phys. {\bf 19}, 941-965 (2007).

\bibitem{ITA} M.~Inui, S.~A.~Trugman, E.~Abrahams, {\sl Unusual properties of midband states in systems with off-diagonal disorder}, Phys. Rev. {\bf B49}, 3190-3196 (1994).

\bibitem{PaS} S.~Palpacelli, S.~Succi, {\sl Numerical Evidence of Sinai Diffusion of Random-Mass Dirac Particles}, Commun. Comput. Phys. {\bf 23},  899-909 (2018).

\bibitem{PS} E.~Prodan, H.~Schulz-Baldes, {\sl Bulk and Boundary Invariants for Complex Topological Insulators: From $K$-Theory to Physics}, (Springer International, Cham, 2016).

\bibitem{Ran} N.~Rangamani,  {\sl Singular-unbounded random Jacobi matrices}, J. Math. Phys. {\bf 60}, 081904 (2019).

\bibitem{SSt1} H.~Schulz-Baldes, T.~Stoiber, {\sl Harmonic analysis in operator algebras and its applications to index theory and topological solid state systems}, (Springer International, Cham, 2022).

\bibitem{Sha} J. Shapiro, {\sl Incomplete localization for disordered chiral strips},  J. Math. Phys. {\bf 64}, 081902 (2023).

\bibitem{SH} K.~Somnatha, R.~Herbei, {\sl Joint exact simulation of the maximum and its location for
some constrained Brownian processes}, J. Stat. Comp. Simulation {\bf 96}, 2849-2872 (2026).

\bibitem{SE} C.~M.~Soukoulis, E.~N.~Economou, {\sl Off-diagonal disorder in one-dimensional systems},  Phys. Rev. {\bf B 24}, 5698-5701 (1981).

\bibitem{SSH} W.~P.~Su, J.~R.~Schrieffer, A.~J.~Heeger, {\sl Soliton excitations in polyacetylene}, Phys. Rev. {\bf B 22}, 2099-2111 (1980).

\bibitem{TC} G.~Theodorou, M.~H.~Cohen, {\sl Extended states in a one-demensional system with off-diagonal disorder}, Phys. Rev. {\bf B 13}, 4597-4601 (1976).

\bibitem{VV} B.~Valk\'o, B.~Vir\'ag, {\sl Random Schr\"odinger operators on long boxes, noise explosion and the GOE}, Trans. AMS {\bf  366}, 3709-3728 (2014).

\end{thebibliography}
\end{document}